\documentclass[journal]{IEEEtran}

\usepackage{cite}
\usepackage{array}
\usepackage{color}
\usepackage{float}
\usepackage{stfloats}
\usepackage[cmex10]{amsmath}
\usepackage{amsthm}
\usepackage{nomencl}						
\usepackage[normalem]{ulem}			      
\usepackage{diagbox}			 			
\usepackage{colortbl}						
\usepackage{multirow}						
\usepackage{tabularx}
\usepackage{placeins}
\usepackage{algorithm}
\usepackage{mathrsfs}
\usepackage{mathdots}
\usepackage{amssymb}
\usepackage{arydshln}
\usepackage{textcomp,gensymb}
\usepackage[noend]{algpseudocode}
\usepackage{bm}
\usepackage{url}
\usepackage[T1]{fontenc}
\usepackage{cuted}
\usepackage{geometry}
\usepackage{amsmath}
\usepackage{xcolor}
\usepackage{xpatch}
\usepackage{wrapfig}
\usepackage{makecell}
\usepackage{fancyhdr} 

\ifCLASSINFOpdf
  \usepackage[pdftex]{graphicx}
  \DeclareGraphicsExtensions{.pdf,.jpeg,.png}
\else
  \usepackage[dvips]{graphicx}
  \DeclareGraphicsExtensions{.eps}
\fi

\ifCLASSOPTIONcompsoc
  \usepackage[caption=false,font=normalsize,labelfont=sf,textfont=sf]{subfig}
\else
  \usepackage[caption=false,font=footnotesize]{subfig}
\fi

\newcommand{\figref}[1]{\figurename~\ref{#1}}

\newcommand{\diag}{\text{diag}}

\ifCLASSOPTIONonecolumn

\else

\renewcommand{\arraystretch}{1.5}
\fi

\makeatletter
\ExplSyntaxOn
\cs_new:Npn \bibColoredItems #1#2
  {
    \clist_map_inline:nn {#2} { \cs_new:cpn {bib@colored@##1} {#1} } 
  }
\ExplSyntaxOff

\newcommand\bib@setcolor[1]{%
  \ifcsname bib@colored@#1\endcsname
    \expanded{\noexpand\color{\csname bib@colored@#1\endcsname}}%
  \else
    \normalcolor
  \fi
}

\xpatchcmd\@bibitem
  {\item}
  {\bib@setcolor{#1}\item}
  {}{\fail}

\xpatchcmd\@lbibitem
  {\item}
  {\bib@setcolor{#2}\item}
  {}{\fail}
\makeatother

\newtheorem{theorem}{Theorem} 
\newtheorem{lemma}[theorem]{Lemma} 
 
\newtheorem{proposition}[theorem]{Proposition} 
\theoremstyle{definition} 
\newtheorem{definition}[theorem]{Definition} \theoremstyle{remark} 
\newtheorem{remark}[theorem]{Remark}

    \emergencystretch=\maxdimen
\begin{document}
\bstctlcite{IEEEexample:BSTcontrol}
\setcounter{page}{1}
\title{A Small-Gain-Like Framework for Large-Signal Stability Evaluation of Multi-Converter Systems}
\author
{
    Qiannan~Qu, \IEEEmembership{Student Member}, \IEEEmembership{IEEE}, 
    Kaiwen~Chen, \IEEEmembership{Member}, \IEEEmembership{IEEE}, 
    \\Xin~Xiang, \IEEEmembership{Member}, \IEEEmembership{IEEE},
    Wuhua~Li, \IEEEmembership{Senior Member}, \IEEEmembership{IEEE}, and
    Yunjie~Gu, \IEEEmembership{Senior Member}, \IEEEmembership{IEEE} \vspace{-2em}
    \thanks{This work has been submitted to the IEEE for possible publication.
    Copyright may be transferred without notice, after which this version may no longer be accessible.}
    \vspace{-0.5\baselineskip}
    }
\maketitle
\thispagestyle{fancy}
\renewcommand{\headrulewidth}{0pt}
\begin{abstract}
The increasing penetration of grid-connected converters has greatly altered the large-signal behavior of power systems. Their angle dynamics, shaped by diverse control algorithms and coupled through complex circuit interactions, pose substantial challenges to large-signal stability evaluation of multi-converter systems.
To resolve this issue, the small-gain theorem, which characterizes the dissipation capability of interconnected systems (the small-gain-like property) via the individual dissipation capabilities of subsystems, is introduced to investigate transient angle motions in multi-converter systems. 
A large-signal model involving a set of interconnected relative angle motions is first developed, and the small-gain-like property is then established for multi-converter systems within certain angle limits, which further enables the construction of Lyapunov functions. 
Based on this, an ellipsoidal forward-invariant region is identified inside the angle-limit region, which serves as an effective estimate of the large-signal stability region for multi-converter systems. 
The method is further applied to a paralleled system and a four-converter system, where the large-signal stability boundaries are explicitly computed and subsequently validated through experiments.
The proposed small-gain-like based large-signal stability evaluation method enables quantitative stability assessment in multi-converter systems with diverse control algorithms, which may provide a scalable framework for large-signal stability evaluation and parameter design in modern power systems.

\end{abstract}
\vspace{-5pt}
\section*{Nomenclature}
\addcontentsline{toc}{section}{Nomenclature}
\begin{IEEEdescription}[\IEEEusemathlabelsep\IEEEsetlabelwidth{$i1,  y1$}]
\item[$\theta_{li}$] Power angle of the $i$th GFL converter
\item[$\theta_{mi}$] Power angle of the $i$th GFM converter
\item[$I_{li}e^{j\varphi_{li}}$] Current reference vector of the $i$th GFL converter
\item[$U_{mi}$] Voltage reference of the $i$th GFM converter
\item[$U_{lqi}$] Q-axis voltage component of the $i$th GFL converter
\item[$P_{emi}$] Active power output of the $i$th GFM converter   
\item[$P_{0mi}$] Active power reference of the $i$th GFM converter   
\item[$K_{pi}$] PLL proportional gain of the $i$th GFL converter
\item[$D_{pi}$] Droop coefficient of the $i$th GFM converter 
\item[$\mathbf{Y}_{net}$] System admittance matrix
\item[$\delta_{ij}^{ll}$] Relative angle between the $i$th GFL converter and the $j$th GFL converter
\item[$\delta_{ij}^{lm}$] Relative angle between the $i$th GFL converter and the $j$th GFM converter
\item[$\delta_{ij}^{mm}$] Relative angle between the $i$th GFM converter and the $j$th GFM converter
\item[$\delta_{ij}^{lle}$] Steady-state relative angle of $\delta_{ij}^{ll}$
\item[$\delta_{ij}^{lme}$] Steady-state relative angle of $\delta_{ij}^{lm}$
\item[$\delta_{ij}^{mme}$] Steady-state relative angle of $\delta_{ij}^{mm}$
\item[$\Delta\delta_{ij}^{lm}$] Relative angle deviation $\delta_{ij}^{lm}-\delta_{ij}^{lme}$
\item[$\Delta\delta_{ij}^{mm}$] Relative angle deviation $\delta_{ij}^{mm}-\delta_{ij}^{mme}$
\end{IEEEdescription}
\section{Introduction}\label{section_1}

\IEEEPARstart{T}{he} increasing penetration of grid-connected converters is reshaping modern power systems\cite{guPowerSystemStability2022,rocabertControlPowerConverters2012,hatziargyriouDefinitionClassificationPower2021}. In contrast to synchronous generators (SGs), whose dynamics are dictated by the physical laws of electromechanics, converter dynamics are predominantly governed by software-defined control strategies\cite{blaabjergOverviewControlGrid2006,zhongSelfSynchronizedSynchronvertersInverters2014,fangInertiaFutureMoreElectronics2019}. This paradigm shift from physics-governed to control-governed behavior fundamentally reshapes system dynamics 
and introduces profound challenges for large-signal stability analysis of systems containing multiple converters\cite{wangGridSynchronizationStabilityConverterBased2020,choopaniNewTransientStability2020,fuLargeSignalStabilityGridForming2021}.

Time-domain simulation has long been heavily relied on for large-signal stability studies\cite{kundurPowerSystemStability1994,nagelHighSpeedPowerSystem2013}. However, this approach is highly time-consuming, especially for converter-dominated systems that require smaller timesteps and a larger number of simulation scenarios to accurately capture their fast and nonlinear dynamics, making it unsuitable for real-time stability assessment\cite{anghelAlgorithmicConstructionLyapunov2013,vuLyapunovFunctionsFamily2016a}. 
Alternatively, efforts have been made to extend the direct methods, which are widely used for large-signal stability region estimation in SG-based systems\cite{chiangDirectMethodsStability2011,kundurPowerSystemStability1994}, to converter-dominated systems for rapid large-signal stability evaluation without time-domain simulations. 
Among them, the equal-area criterion (EAC) has been well extended to converters in a single-machine-infinite-bus (SMIB) setup, where the converter operates under either grid-following (GFL) or grid-forming (GFM) control structures\cite{wuDesignOrientedTransientStability2020,liIterativeEqualArea2023,leiQuantitativeIntuitiveVSG2023,liDampingTurningRule2023}. 
It has been shown that the power angle dynamics of both GFL and GFM converters can be reformulated into a swing-like model, to which EAC is applicable\cite{fuLargeSignalStabilityGridForming2021,maGeneralizedSwingEquation2022}. 
However, this methodology faces fundamental changes when scaling up to multi-converter systems. Due to the presence of nonuniform damping terms introduced by GFM converters and cosine-type interaction terms by GFL converters in the angle dynamics\cite{leiDestabilizingMechanismNonuniform2024,zhangInteractionTransientStability2025a}, the global energy function no longer exists for systems containing multiple GFL and GFM converters.

{\color{black}
To analyze the large-signal stability of systems with multiple converters, several attempts have been made to approximate or transform multi-converter systems into a set of equivalent SMIB systems by aggregating converters with similar dynamic characteristics into groups\cite{taulReducedOrderAggregatedModeling2021,palReducedOrderModelingTransient2023,chenAggregatedModelVirtual2021,shenLyapunovMethodBasedCoherent2025}. With the simplified SMIB models, the EAC method can be further extended for multi-converter system stability assessment\cite{xueExtendedEqualArea1989}.
However, converter aggregation relies on the identification of coherent angle dynamics, which is usually performed based on specific circuit topologies\cite{taulReducedOrderAggregatedModeling2021,palReducedOrderModelingTransient2023} or uniform control structures\cite{chenAggregatedModelVirtual2021,shenLyapunovMethodBasedCoherent2025}, and thus may fail to capture the dominant dynamics of large-scale systems or provide effective stability assessment.
}
In addition, algebraic reformulation techniques of Lyapunov functions have been developed in SG-dominated power systems to characterize their large-signal stability performance\cite{anghelAlgorithmicConstructionLyapunov2013,vuLyapunovFunctionsFamily2016a}. Nevertheless, these methods may face challenges in handling the various nonlinear terms inherent in multi-converter systems.

Therefore, a new mathematical tool that allows rigorous and scalable large-signal stability analysis for multi-converter systems is desired. To this end, this paper introduces a \textit{small-gain-like framework} from nonlinear control theory to provide a new pathway to solve the large-signal stability problem \cite{yMonotoneStabilityNonlinear1992,jiangSmallgainTheoremISS1994}. 
The small-gain theorem is a well-known result in robust control. It was originally proposed for standalone systems but was recently generalized to interconnected systems\cite{dashkovskiyISSSmallGain2007,jiangGeneralizationNonlinearSmallgain2008,chenActiveNodesNetwork2024}. Unlike the energy-based methods including EAC, the small-gain-like framework does not rely on the physical property (like energy conservation) of the system and is therefore very suitable to the control-governed behaviors of converters. 
The small-gain theorem states that a system consisting of two subsystems connected in a feedback loop is stable if each subsystem holds a certain property and the loop gain is less than one. Since the loop gain can be defined via an "induced norm", the small-gain theorem is suitable for both linear and nonlinear systems. 
On this basis, the theorem has been further extended to large-scale interconnected systems\cite{dashkovskiyISSSmallGain2007,jiangGeneralizationNonlinearSmallgain2008}, which can be intuitively interpreted as: if each subsystem possesses a certain stability property under the influence of other connected subsystems and their interactions in between do not amplify each other, then the same stability property can be extended to the entire system. 
In this manner, the dissipation capability of an interconnected system, referred to as the \textit{small-gain-like property}, can be derived from the individual dissipation capabilities of subsystems and their interconnections, 
which in turn enables the construction of a sum-type Lyapunov function that implies the overall stability properties\cite{chenActiveNodesNetwork2024}. 

In this paper, the small-gain-like framework is introduced to evaluate the large-signal stability of multi-converter systems containing different types of converters. 
The power angle motion of each converter is modeled, and the overall angle dynamics are fully captured by a set of interconnected relative angle motions. 
Based on this, the individual dissipation capabilities of relative angle motions among converters with different control structures are mathematically depicted, and the small-gain-like property of the whole system is established within certain angle limits, which provides a criterion for constructing a Lyapunov function.
Moreover, an ellipsoidal stability region within the angle-limit set is identified, which is forward invariant and guarantees convergence of all states inside it to the equilibrium. 
As a result, the proposed method can directly solve the large-signal stability boundary of multi-converter systems with diverse control algorithms and various circuit structures, and its flexible framework shows a strong potential for extension to larger systems for stability assessment and control optimization. This, in turn, provides structural insights into system strengths and vulnerabilities from the perspective of large-signal stability. 
The rest of this paper is organized as follows. Section~\ref{section_2} develops a large-signal model of multi-converter systems that captures the interconnected relative angle dynamics. 
In Section~\ref{section_3}, the small-gain-like property of multi-converter systems is investigated, and an ellipsoidal stability region is derived for large-signal stability assessment. 
The proposed method is further applied to evaluate the large-signal stability behavior of a paralleled system and a four-converter system in Section~\ref{section_4}, and the results are further validated through experiments in Section~\ref{section_5}. Finally, Section~\ref{section_6} concludes this paper.

\section{System Description and Modeling}\label{section_2}
This section first introduces the circuit and control structures of the multi-converter systems. On top of that, the relative angle motions between different types of converters are separately modeled, and a large-signal model describing the overall angle dynamics of the whole system is proposed.

\begin{figure}\centering
    \includegraphics[scale=0.92]{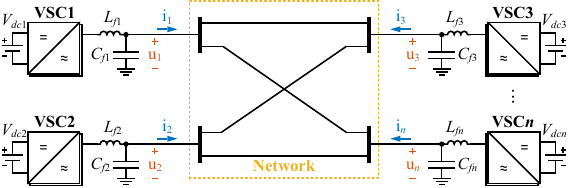}
    \vspace{-2pt}
    \caption{Generalized topology of a multi-converter system}
    \label{Fig1_Generalized}
    \vspace{-2pt}    
\end{figure}

\begin{figure}\centering
    \includegraphics[scale=1.1]{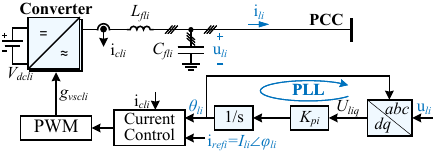}\\
    \footnotesize (a)\\
    \vspace{-2pt}
    \includegraphics[scale=1.1]{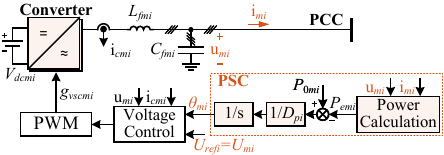}\\
    \footnotesize (b)\\
    \vspace{-2pt}
    \caption{Control structures of converters. (a) GFL converter. (b) GFM converter}
    \label{Fig2_Control}
    \vspace{-5pt}
\end{figure}

\vspace{-5pt}
\subsection{System Description}
The generalized circuit structure of a system containing $n$ converters is illustrated in \figref{Fig1_Generalized}. Each converter connects to the network through an LC filter, and the terminal voltage and current at the $i$th converter node are denoted by $\mathrm{u}_i$ and $\mathrm{i}_i$, respectively.
{\color{black}
The converters in the system shown in \figref{Fig1_Generalized} are classified into GFL converters and GFM converters according to their distinct control structures. For the GFL converter, it synchronizes with its terminal voltage through a first-order PLL structure and tracks the constant current reference $\mathrm{i}_{ref}$ via the current control loop, as shown in \figref{Fig2_Control}(a). In this manner, it can realize precise power delivery and rapid post-fault power recovery.
For the GFM converter in \figref{Fig2_Control}(b), it generates the terminal voltage angle reference using the power synchronization control (PSC) algorithm, and maintains a constant voltage amplitude $U_{ref}$ via the inner voltage control loop.
Therefore, the GFM converter exhibits voltage-source characteristics and can provide strong voltage support for the overall system.
}
{\color{black}
Considering recent developments in transient over-current capability enhancement technologies\cite{liuStabilityConstrainedTransient2024,jiangITjReducedPower2021}, power electronic converters have gradually gained the capability to withstand certain large currents over short timescales. Therefore, GFM converters in multi-converter systems are assumed to maintain their normal operating mode under large disturbances, such that their intrinsic synchronization dynamics can be clearly characterized. Moreover, the same modeling method also has the potential to be extended to some current-limiting operating scenarios by modifying the corresponding circuit relationships and parameters\cite{shuaiTransientAngleStability2019,wangTransientStabilityAnalysis2025}.
}
Considering that the inner voltage and current loops typically respond much faster than the PLL and the PSC dynamics\cite{leiQuantitativeIntuitiveVSG2023,wuDesignOrientedTransientStability2019}, the GFL output current and the GFM terminal voltage are assumed to perfectly track their references. 
In this way, assuming that the vector current reference of the $i$th GFL converter in its own PLL-defined frame is $I_{li}e^{j\varphi_{li}}$ and the voltage amplitude reference of the $i$th GFM converter is $U_{mi}$, the output current of the $i$th GFL converter, $\mathrm{i}_{li}$, and the terminal voltage of the $i$th GFM converter, $\mathrm{u}_{mi}$, are determined as
\begin{equation}
   \mathrm{i}_{li}=I_{li}e^{j(\theta_{li}+\varphi_{li})}\quad\mathrm{u}_{mi}=U_{mi}e^{j\theta_{mi}}
\end{equation}
where $\theta_{li}$ and $\theta_{mi}$ are the power angles of the $i$th GFL converter and the $i$th GFM converter.

Define the current and voltage vectors of the GFL converters as $\mathbf{i}_{l}$ and $\mathbf{u}_l$, and those of the GFM converters as $\mathbf{i}_m$ and $\mathbf{u}_m$. For a system with $p$ GFL converters and $q$ GFM converters, the current and voltage vectors are represented as
\begin{equation}
    \begin{aligned}
        \mathbf{i}_l&=[\mathrm{i}_{l1},\dots,\mathrm{i}_{lp}]^{\mathrm{T}}\quad\quad\ \mathbf{u}_l=[\mathrm{u}_{l1},\dots,\mathrm{u}_{lp}]^{\mathrm{T}}\\[-2pt]
        \mathbf{i}_m&=[\mathrm{i}_{m1},\dots,\mathrm{i}_{mq}]^{\mathrm{T}}\quad\mathbf{u}_m=[\mathrm{u}_{m1},\dots,\mathrm{u}_{mq}]^{\mathrm{T}}.
    \end{aligned}
\end{equation}

Moreover, define the system admittance matrix as $\mathbf{Y}_{net}$, the current and voltage vectors satisfy
\begin{equation}\label{equ_Circuit_MM}
    \begin{bmatrix}
    \mathbf{i}_l\\[-2pt]
    \mathbf{i}_m
    \end{bmatrix}
    =
    \underbrace{\begin{bmatrix}
    \mathbf{Y}_{ll} & \mathbf{Y}_{lm} \\[-2pt]
    \mathbf{Y}_{ml} & \mathbf{Y}_{mm}
    \end{bmatrix}}_{\mathbf{Y}_{net}}
    \begin{bmatrix}
    \mathbf{u}_l \\[-2pt]
    \mathbf{u}_m
    \end{bmatrix}.
\end{equation}

With \eqref{equ_Circuit_MM}, the terminal voltage of GFL converters and the output current of GFM converters are calculated as
\begin{equation}
    \begin{aligned}
        \mathbf{u}_{l}=&\mathbf{Y}_{ll}^{-1}\left(\mathbf{i}_{l}-\mathbf{Y}_{lm}\mathbf{u}_{m}\right)\\[-2pt]
        \mathbf{i}_{m}=&\mathbf{Y}_{ml}\mathbf{Y}_{ll}^{-1}\mathbf{i}_{l}
                        +\left(\mathbf{Y}_{mm}-\mathbf{Y}_{ml}\mathbf{Y}_{ll}^{-1}\mathbf{Y}_{lm}\right)\mathbf{u}_m.
    \end{aligned}
\end{equation}

{\color{black}
Based on the network reduction and current-voltage relationships, the electrical quantities adopted for synchronization control can be derived. 
The quadrature voltage component of the $i$th GFL converter and the active power of the $i$th GFM converter are expressed as follows
\begin{equation}\label{equ_UqPe_org}
    \begin{aligned}
    &U_{lqi}=\Im{(\mathrm{u}_{li}e^{-j\theta_{li}})}=\Im{\left\{\sum_{k=1}^{p}\left[(\mathbf{Y}_{ll}^{-1})_{ik}I_{lk}e^{j(\theta_{lk}-\theta_{li}+\varphi_{lk})}\right]\right\}}\\
    &-\Im{\left\{\sum_{r=1}^{q}\left[(\mathbf{Y}_{ll}^{-1}\mathbf{Y}_{lm})_{ir}U_{mr}e^{j(\theta_{mr}-\theta_{li})}\right]\right\}}\\
    &P_{emi}=\Re{(\mathrm{u}_{mi}\mathrm{i}_{mi}^*)}\\
           &=\Re\left\{\sum_{k=1}^{p}\left[(\mathbf{Y}_{ml}\mathbf{Y}_{ll}^{-1})_{ik}U_{mi}I_{lk}e^{j(\theta_{lk}-\theta_{mi}+\varphi_{lk})}\right]\right\}\\
          &+\Re{\left\{\sum_{r=1}^{q}\left[(\mathbf{Y}_{mm}-\mathbf{Y}_{ml}\mathbf{Y}_{ll}^{-1}\mathbf{Y}_{lm})_{ir}U_{mi}U_{mr}e^{j(\theta_{mr}-\theta_{mi})}\right]\right\}}.
    \end{aligned}
\end{equation}

For notational simplicity, circuit parameters are defined as
\begin{equation}
  \begin{aligned}
    &Z_{ik}^{ll}e^{j\phi_{ik}^{ll}}=(\mathbf{Y}_{ll}^{-1})_{ik}\quad Y_{ir}^{mm}e^{j\phi_{ir}^{mm}}=(\mathbf{Y}_{mm}-\mathbf{Y}_{ml}\mathbf{Y}_{ll}^{-1}\mathbf{Y}_{lm})_{ir}\\
    &G_{ir}^{lm}e^{j\phi_{ir}^{lm}}=(\mathbf{Y}_{ll}^{-1}\mathbf{Y}_{lm})_{ir}\quad G_{ik}^{ml}e^{j\phi_{ik}^{ml}}=(\mathbf{Y}_{ml}\mathbf{Y}_{ll}^{-1})_{ik}.
  \end{aligned}
\end{equation}

Accordingly, the expressions in \eqref{equ_UqPe_org} can be rewritten as 
\begin{equation}\label{equ_UqPe}
  \begin{aligned}
    U_{lqi}&=\sum\nolimits_{k=1}^p\left[Z_{ik}^{ll}I_{lk}\sin{\left(\theta_{lk}-\theta_{li}+\varphi_{lk}+\phi_{ik}^{ll}\right)}\right]\\
    &-\sum\nolimits_{r=1}^{q}\left[G^{lm}_{ir}U_{mr}\sin{\left(\theta_{mr}-\theta_{li}+\phi_{ir}^{lm}\right)}\right]\\
    P_{emi}&=\sum\nolimits_{k=1}^{p}\left[G^{ml}_{ik}U_{mi}I_{lk}\cos{\left(\theta_{lk}-\theta_{mi}+\varphi_{lk}+\phi_{ik}^{ml}\right)}\right]\\
    &+\sum\nolimits_{r=1}^{q}\left[Y^{mm}_{ir}U_{mi}U_{mr}\cos{\left(\theta_{mr}-\theta_{mi}+\phi_{ir}^{mm}\right)}\right].
  \end{aligned}
\end{equation}
}
On this basis, the large-signal angle motions of individual converters are mathematically described by \eqref{equ_AngleDynamics}, where $K_{pi}$ is the PLL proportional gain of the $i$th GFL converter, and $D_{pi}$ and $P_{0mi}$ denote the droop coefficient and the power reference of the $i$th GFM converter.
\begin{equation}\label{equ_AngleDynamics}
  \begin{aligned}
  \dot\theta_{li}=K_{pi}U_{lqi}\quad\quad
  \dot\theta_{mi}=\left(P_{0mi}-P_{emi}\right)/D_{pi}
  \end{aligned}
\end{equation}

\vspace{-5pt}
\subsection{Relative Angle Motion Characterization}
In this subsection, the relative angle motions between different types of converters are modeled by combining \eqref{equ_UqPe} and \eqref{equ_AngleDynamics}. The relative angles $\delta_{ij}^{ll}$ between two GFL converters, $\delta_{ij}^{lm}$ between a GFL converter and a GFM converter, and $\delta_{ij}^{mm}$ between two GFM converters are defined as
\begin{equation}
    \delta_{ij}^{ll}=\theta_{li}-\theta_{lj}\quad
    \delta_{ij}^{lm}=\theta_{li}-\theta_{mj}\quad
    \delta_{ij}^{mm}=\theta_{mi}-\theta_{mj}.
\end{equation}

{\color{black}
The relative angle motions for $\delta_{ij}^{ll}$, $\delta_{ij}^{lm}$ and $\delta_{ij}^{mm}$ are then mathematically described, as shown in \eqref{equ_Relative_LL} to \eqref{equ_Relative_MM}.
From the derived relative angle motion equations, it is evident that the dynamics of any relative angle between two converters are influenced by all other relative angles in the system. Consequently, the angle dynamics of a multi-converter system form a large-scale interconnected structure with complex interaction characteristics, which makes intuitive and quantitative large-signal stability assessment challenging. 
Moreover, in a system containing $n$ converters, there are $n(n-1)/2$ relative angle motions between converter pairs, while the overall large-signal dynamics can be described by a nonlinear equation of order $(n-1)$. This implies that many of these relative motions are not independent but are constrained by others, which requires further simplification. 
}

\begin{equation}\label{equ_Relative_LL}
  \begin{aligned}
    \dot\delta_{ij}^{ll}&=K_{pi}Z^{ll}_{ij}I_{lj}\sin{\left(-\delta_{ij}^{ll}+\varphi_{lj}+\phi_{ij}^{ll}\right)}\\
    &-K_{pj}Z^{ll}_{ji}I_{li}\sin{\left(\delta_{ij}^{ll}+\varphi_{li}+\phi_{ji}^{ll}\right)}\\
    &+K_{pi}\sum\nolimits_{k=1,k\neq j}^p\left[Z^{ll}_{ik}I_{lk}\sin{\left(\delta_{ki}^{ll}+\varphi_{lk}+\phi_{ik}^{ll}\right)}\right]\\
    &-K_{pi}\sum\nolimits_{r=1}^{q}\left[G^{lm}_{ir}U_{mr}\sin{\left(-\delta_{ir}^{lm}+\phi_{ir}^{lm}\right)}\right]\\
    &-K_{pj}\sum\nolimits_{k=1,k\neq i}^p\left[Z^{ll}_{jk}I_{lk}\sin{\left(\delta_{kj}^{ll}+\varphi_{lk}+\phi_{jk}^{ll}\right)}\right]\\
    &+K_{pj}\sum\nolimits_{r=1}^{q}\left[G^{lm}_{jr}U_{mr}\sin{\left(-\delta_{jr}^{lm}+\phi_{jr}^{lm}\right)}\right]
  \end{aligned}
\end{equation}

\begin{equation}\label{equ_Relative_LM}
  \begin{aligned}
    \dot\delta_{ij}^{lm}&=K_{pi}G_{ij}^{lm}U_{mj}\sin{\left(\delta_{ij}^{lm}-\phi_{ij}^{lm}\right)}\\
    &+G_{ji}^{ml}U_{mj}I_{li}\cos{\left(\delta_{ij}^{lm}+\varphi_{li}+\phi_{ji}^{ml}\right)}/D_{pj}-P_{0mj}/D_{pj}\\
    &+K_{pi}\sum\nolimits_{k=1}^p\left[Z^{ll}_{ik}I_{lk}\sin{\left(\delta_{ki}^{ll}+\varphi_{lk}+\phi_{ik}^{ll}\right)}\right]\\
    &+K_{pi}\sum\nolimits_{r=1,r\neq j}^{q}\left[G^{lm}_{ir}U_{mr}\sin{\left(\delta_{ir}^{lm}-\phi_{ir}^{lm}\right)}\right]\\
    &+\sum\nolimits_{k=1,k\neq i}^{p}\left[G_{jk}^{ml}U_{mj}I_{lk}\cos{\left(\delta_{kj}^{lm}+\varphi_{lk}+\phi_{jk}^{ml}\right)}\right]/D_{pj}\\
    &+\sum\nolimits_{r=1}^q\left[Y_{jr}^{mm}U_{mj}U_{mr}\cos{\left(\delta_{rj}^{mm}+\phi_{jr}^{mm}\right)}\right]/D_{pj}
  \end{aligned}
\end{equation}

\begin{equation}\label{equ_Relative_MM}
  \begin{aligned}
    \dot\delta_{ij}^{mm}&=-Y_{ij}^{mm}U_{mi}U_{mj}\cos{\left(\delta_{ij}^{mm}-\phi_{ij}^{mm}\right)}/D_{pi}\\
    &+Y_{ji}^{mm}U_{mi}U_{mj}\cos{\left(\delta_{ij}^{mm}+\phi_{ji}^{mm}\right)}/D_{pj}\\
    &-\sum\nolimits_{k=1}^p\left[G_{ik}^{ml}U_{mi}I_{lk}\cos{\left(\delta_{ki}^{lm}+\varphi_{lk}+\phi_{ik}^{ml}\right)}\right]/D_{pi}\\
    &-\sum\nolimits_{r=1,r\neq j}^q\left[Y_{ir}^{mm}U_{mi}U_{mr}\cos{\left(\delta_{ri}^{mm}+\phi_{ir}^{mm}\right)}\right]/D_{pi}\\
    &+\sum\nolimits_{k=1}^p\left[G_{jk}^{ml}U_{mj}I_{lk}\cos{\left(\delta_{kj}^{lm}+\varphi_{lk}+\phi_{jk}^{ml}\right)}\right]/D_{pj}\\
    &+\sum\nolimits_{r=1,r\neq i}^q\left[Y_{jr}^{mm}U_{mj}U_{mr}\cos{\left(\delta_{rj}^{mm}+\phi_{jr}^{mm}\right)}\right]/D_{pj}\\    
    &+P_{0mi}/D_{pi}-P_{0mj}/D_{pj}
  \end{aligned}
\end{equation}

\vspace{-5pt}
\subsection{Large-signal Modeling of the Multi-converter System}\label{section_2c}
The self-damping effects of each type of relative motion are represented by the first two terms in \eqref{equ_Relative_LL} to \eqref{equ_Relative_MM}. 
{\color{black}
In practical power networks, transmission lines are predominantly inductive, and thus the circuit angles satisfy $\phi_{ij}^{ll} = \phi_{ji}^{ll} \approx \frac{\pi}{2}, \ 
\phi_{ij}^{lm} = \phi_{ji}^{ml} \approx \pi, \ 
\phi_{ij}^{mm} = \phi_{ji}^{mm} \approx \frac{\pi}{2}$. Moreover, according to the symmetry of the admittance matrix $\mathbf{Y}_{net}$, the magnitudes of the corresponding coupling coefficients are equal, i.e., $|Z_{ij}^{ll}| = |Z_{ji}^{ll}|$, $|G_{ij}^{lm}| = |G_{ji}^{ml}|$, $|Y_{ij}^{mm}| = |Y_{ji}^{mm}|$. As a result, under typical operating conditions where GFL converters inject predominantly active current into the network for efficient power delivery, 
the self-damping terms can be approximated as 
\begin{equation}\label{equ_SelfDamp_Approx}
\begin{aligned}
f_{ij}^{ll}(\delta_{ij}^{ll})&\approx Z^{ll}_{ij}\cos\delta_{ij}^{ll}(K_{pi}I_{lj}-K_{pj}I_{li})\\
f_{ij}^{lm}(\delta_{ij}^{lm})&\approx -K_{pi}G_{ij}^{lm}U_{mj}\sin\delta_{ij}^{lm}
    				   -\frac{G_{ij}^{lm}U_{mj}I_{li}\cos\delta_{ij}^{lm}}{D_{pj}}\\
f_{ij}^{mm}(\delta_{ij}^{mm})&\approx -Y_{ij}^{mm}U_{mi}U_{mj}\sin{\delta_{ij}^{mm}}(\frac{1}{D_{pi}}+\frac{1}{D_{pj}})
\end{aligned}
\end{equation}
where $f_{ij}^{ll}$, $f_{ij}^{lm}$ and $f_{ij}^{mm}$ represent the self-damping terms corresponding to GFL-GFL, GFL-GFM and GFM-GFM relative angle motions, respectively.
}

{\color{black}
According to \eqref{equ_SelfDamp_Approx}, the two self-damping terms in GFL-GFL relative motion have opposite signs, thereby producing only a limited self-damping effect for the relative motions described in \eqref{equ_Relative_LL}. 
In contrast, for the relative motion between GFM converters, the two self-damping terms act in the same direction, which indicates that the PSC algorithms of both GFM converters help them reach a common speed, leading to a strong self-damping effect for the relative motions in \eqref{equ_Relative_MM}.
Moreover, for the relative motion between GFL and GFM converters in \eqref{equ_Relative_LM}, the self-damping terms generated by the PLL and PSC algorithms are approximately in quadrature and therefore do not cancel each other. Hence, the relative motions in \eqref{equ_Relative_LM} exhibit a stronger self-damping effect compared to those in \eqref{equ_Relative_LL}.
}

{\color{black}
Based on the above analysis of the self-damping effects in different types of relative angle motions, a subset of GFL-GFM and GFM-GFM relative motions is adopted to characterize the large-signal behavior of the entire system. 
To account for the angle dynamics of all converters, for a system with $n$ converters consisting of $p$ GFL converters and $q$ GFM converters, $p$ groups of GFL-GFM relative angle variables are selected to include all GFL converters, and $(q-1)$ groups of GFM-GFM relative angle variables are selected to connect all GFM converters without redundancy.

Among them, $p$ groups of relative motions are selected between each GFL converter and its nearest GFM converter. The nearest GFM converter is defined as the one with the "maximum" self-damping amplitude $A_{ij}^{lm}$ among all GFM converters, where $A_{ij}^{lm}$ is defined by
\begin{equation}\label{equ_Alm}
  \begin{aligned}
    A_{ij}^{lm}&=G^{lm}_{ij}U_{mj}\left[\left(K_{pi}\cos\phi_{ij}^{lm}-I_{li}\sin{\left(\varphi_{li}+\phi_{ij}^{lm}\right)}/D_{pj}\right)^2\right.\\[-2pt]
    &\left.+\left(-K_{pi}\sin\phi_{ij}^{lm}+I_{li}\cos{\left(\varphi_{li}+\phi_{ij}^{lm}\right)}/D_{pj}\right)^2\right]^{\frac{1}{2}}.
  \end{aligned}
\end{equation}

Moreover, a weighted graph is constructed using the GFM converters as graph nodes, where the edge weights are determined by the corresponding self-damping amplitudes, which can be calculated through
\begin{equation}\label{equ_Amm}
  \begin{aligned}
    A_{ij}^{mm}&=Y_{ij}^{mm}U_{mi}U_{mj}\left[\left(\cos\phi_{ij}^{mm}\right)^2\left(1/D_{pi}-1/D_{pj}\right)^2\right.\\[-2pt]
    &+\left.\left(\sin\phi_{ij}^{mm}\right)^2\left(1/D_{pi}+1/D_{pj}\right)^2\right]^{\frac{1}{2}}.
  \end{aligned}
\end{equation}
Based on this, the selected $(q-1)$ GFM-GFM relative motions are determined according to the maximum spanning tree of the weighted graph, such that all GFM converters are connected without redundancy while preserving relatively strong self-damping capability.

\begin{lemma}
Consider a system with $n=p+q$ converters. Let the selected set of relative angle variables consist of:

1) $p$ GFL-GFM relative motions, each connecting one GFL converter to one GFM converter;

2) $(q-1)$ GFM-GFM relative motions selected according to the maximum spanning tree of a weighted graph, where the $q$ GFM converters are treated as graph nodes and the corresponding self-damping amplitudes $A_{ij}^{mm}$ are adopted as the edge weights.

Then, there exists an invertible transformation matrix $\mathbf{T}_r$ that maps the original system angle states to the selected set of relative angle variables.
\end{lemma}

\begin{proof}
Select the $q$th GFM converter as the reference converter and define the conventional reference-based relative angle deviation vector as\cite{chiangDirectMethodsStability2011}
\begin{equation}
\bm{\eta}=
\begin{bmatrix}
\Delta\theta_{l1}-\Delta\theta_{mq}\\[-2pt]
\vdots\\[-2pt]
\Delta\theta_{lp}-\Delta\theta_{mq}\\[-2pt]
\Delta\theta_{m1}-\Delta\theta_{mq}\\[-2pt]
\vdots\\[-2pt]
\Delta\theta_{m(q-1)}-\Delta\theta_{mq}
\end{bmatrix}
\in\mathbb{R}^{n-1}
\end{equation}
where $\Delta\theta_{(\cdot)}$ denotes the deviation of the corresponding converter angle from its steady-state value. The selected relative-angle vector is defined as
\begin{equation}
\mathbf{x}_{\delta}
=
[x_{\delta1},x_{\delta2},\ldots,x_{\delta(n-1)}]^T
=
\mathbf{T}_r\bm{\eta}
\end{equation}
where $\mathbf{T}_r\in\mathbb{R}^{(n-1)\times(n-1)}$ is the corresponding transformation matrix.
According to the proposed grouping strategy, $\mathbf{T}_r$ can be described in the following block form
\begin{equation}
\mathbf{T}_r=
\begin{bmatrix}
\mathbf{I}_{p} & \mathbf{H}_{p\times(q-1)}\\[-2pt]
\mathbf{0}_{(q-1)\times p} & \mathbf{B}_{(q-1)\times(q-1)}
\end{bmatrix}.
\end{equation}
The first $p$ selected relative angle variables in $\mathbf{x}_{\delta}$ correspond to the $p$ GFL converters in sequence, where each selected relative angle contains only one GFL angle deviation. Therefore, the upper-left block of $\mathbf{T}_r$ is an identity matrix $\mathbf{I}_{p}$, while the upper-right block $\mathbf{H}$ describes the associated GFM converter in each selected GFL-GFM relative angle motion.
For the selected GFM-GFM relative angle motions, since no GFL angle is included, the lower-left block is a zero matrix $\mathbf{0}_{(q-1)\times p}$.
Furthermore, according to the proposed grouping strategy, the selected $(q-1)$ motions are determined from the maximum spanning tree of the weighted graph. Therefore, the lower-right block $\mathbf{B}$ corresponds to the reduced incidence matrix of a tree, which has  full rank $q-1$ and satisfies $\det(\mathbf{B})\neq0$.
Therefore, $\mathbf{T}_r$ forms a block upper triangular structure, and its determinant satisfies
\begin{equation}
\det(\mathbf{T}_r)
=
\det(\mathbf{I}_{p})\det(\mathbf{B})\neq0.
\end{equation}
Thus, $\mathbf{T}_r$ is invertible, which implies that the selected relative angle set is equivalent to the conventional reference-based relative angle coordinates.
\end{proof}
}

Consequently, the large-signal dynamics of a multi-converter system comprising $p$ GFL converters and $q$ GFM converters can be clearly and completely depicted by $p$ GFL-GFM relative motions and $(q-1)$ GFM-GFM relative motions.


\section{Small-Gain-Like Framework for Large-signal Stability Evaluation}\label{section_3}
In this section, the small-gain-like property of interconnected systems is first established based on the individual dissipation capabilities of subsystems and their connections, through which the Lyapunov stability of the overall system can be characterized. 
It is then extended to multi-converter systems, where the dissipation capabilities of relative angle motions are analyzed and the small-gain-like property is verified within certain angle limits. 
Finally, an ellipsoidal forward-invariant region is identified inside the angle-limit region, which provides a straightforward estimate of the large-signal stability region for multi-converter systems.

\vspace{-5pt}
\subsection{Small-gain-like Property of Interconnected Systems}\label{section_3a}
Consider a system $\Sigma$ composed of $n$ subsystems $\Sigma_i\ (i=1,\dots,n)$ with the state vector $\mathbf{x}=[x_1,\dots,x_n]^{\mathrm{T}}$, $ x_i\in \mathbb{R}^{n_i}$. Each subsystem receives $(n-1)$ inputs from the others, and the input from $\Sigma_j$ to $\Sigma_i$ is denoted by $u_{ij}\in\mathbb{R}^{n_{ij}^u}$. Accordingly, the dynamics of $\Sigma_i$ are represented as
\begin{equation}
  \Sigma_i : \dot x_i=f_i(x_i,u_{ij})\quad(j=1,\dots,n,\ j\neq i)
\end{equation}
where $f_i: \mathbb{R}^{n_i+n_{ui}} \to \mathbb{R}^{n_i}\ (n_{ui}=\sum\nolimits_{j=1, j\neq i}^n{n_{ij}^u})$.

{\color{black}
Moreover, assume that each subsystem satisfies the dissipation inequality
\begin{equation}\label{equ_Subsystem_dissipation}
  \dot V_i \leq -a_i\|x_i\|_2^2+\sum\nolimits_{j=1,j\neq i}^n{b_{ij}\|u_{ij}\|_2^2}
\end{equation}
with respect to a positive definite and radially unbounded Lyapunov function $V_i: \mathbb{R}^{n_i} \to \mathbb{R}_+ \quad \text{of class } \mathcal{C}^1$, where $a_i>0$ and $b_{ij}\geq 0$. 
This inequality characterizes the stability of the corresponding subsystem in Lyapunov-function form, where the derivative $\dot V_i$ consists of two parts: an intrinsic dissipation term representing the local damping capability, and a supply-rate term representing the influence of interconnected inputs from other subsystems.
Accordingly, the dissipation inequality in \eqref{equ_Subsystem_dissipation} captures an input-to-state stability-like property of the subsystem $\Sigma_i$\cite{jiangSmallgainTheoremISS1994,hillStabilityNonlinearDissipative1976}.
}

The subsystems are interconnected via $u_{ij}=x_j$. Under this interconnection, the dissipation inequalities are summarized as
\begin{equation}\label{equ_Subsystem_dissipation_matrix}
  \dot{\bar{V}}\leq -\mathbf{E}\boldsymbol{\phi}(\mathbf{x})
\end{equation}
where $\bar{V}\triangleq[V_1,\dots,V_n]^{\mathrm{T}}$ and $\boldsymbol{\phi}(\mathbf{x})\triangleq[\|x_1\|_2^2,\dots,\|x_n\|_2^2]^{\mathrm{T}}$. The matrix $\mathbf{E}$, termed the \textit{dissipation matrix}, reflects the dissipation capabilities of subsystems and their connections\cite{chenActiveNodesNetwork2024}, is represented as
\begin{align}\label{equ_Ematrix_general}
    \mathbf{E}=\diag(a_1,\dots,a_n)-\mathbf{B}
    \text{, where }
    \left(\mathbf{B}\right)_{ij}=
    \left\{
        \begin{aligned}
          &0\ \ \ i=j\\[-2pt]
          &b_{ij}\  i\neq j
        \end{aligned}
    \right. .
\end{align}

To evaluate the stability of the interconnected system $\Sigma$, it is expected to find the dissipation inequality for the whole system. 
Assume that a Lyapunov function of the whole system $V$ exists and it is constructed as
\begin{equation}\label{equ_energy_like_fcn}
  V=\sum\nolimits_{i=1}^n{c_iV_i}=\mathbf{c}^{\mathrm{T}}\bar{V}\text{, where } \mathbf{c}\triangleq[c_1,\dots,c_n]^{\mathrm{T}}>0.
\end{equation}
If the dissipation inequality
\begin{equation}\label{equ_WholeSystem_dissipation}
  \dot V\leq -W(\mathbf{x})\leq 0
\end{equation} 
is satisfied, where $W(\cdot)$ is a positive definite function of $\mathbf{x}$, then Lyapunov stability of the system $\Sigma$ can be concluded. 
{\color{black}
Nevertheless, it is worth noting that the construction of such a Lyapunov function depends on both the dissipation properties and the interconnection structure of the system, and is not guaranteed for arbitrary interconnected systems. For clarity, we define the \textit{small-gain-like property} of an interconnected system as follows.

\begin{definition}[Small-gain-like property]
An interconnected system $\Sigma$ is said to possess the small-gain-like property if there exists a coefficient vector $\mathbf{c}$ with all positive elements such that a sum-type Lyapunov function $V$ of the form \eqref{equ_energy_like_fcn} exists and satisfies the dissipation inequality \eqref{equ_WholeSystem_dissipation}.
\end{definition}

Accordingly, we refer to the \textit{small-gain-like condition} as a condition on the network parameters in the dissipation matrix $\mathbf{E}$ under which the small-gain-like property holds. The following theorem establishes a sufficient condition for the small-gain-like property of interconnected systems.

\begin{theorem}\label{thm_smallgain}
Consider the interconnected system $\Sigma$ with its dissipation matrix $\mathbf{E}$ defined in \eqref{equ_Ematrix_general}. If matrix $\mathbf{E}$ is a nonsingular \textit{M-matrix}, then the system $\Sigma$ possesses the small-gain-like property. 
\end{theorem}

\begin{proof}
The following proof has been previously established in \cite{chenActiveNodesNetwork2024}, and is reformulated here for completeness.
Substituting the dissipation inequality in \eqref{equ_Subsystem_dissipation_matrix} and the detailed expression of $V$ in \eqref{equ_energy_like_fcn} into $\dot{V}$ in \eqref{equ_WholeSystem_dissipation} yields
\begin{equation}\label{equ_WholeSystem_dissipation_v2}
  \dot V=\mathbf{c}^{\mathrm{T}}\dot{\bar{V}}\leq-\mathbf{c}^{\mathrm{T}} \mathbf{E}\boldsymbol{\phi}(\mathbf{x})\leq 0.
\end{equation}
Since all off-diagonal entries of $\mathbf{E}$ are nonpositive, $\mathbf{E}$ is a \textit{Z-matrix} \cite{plemmonsMmatrixCharacterizationsInonsingularMmatrices1977a}. 
If $\mathbf{E}$ is further a nonsingular \textit{M-matrix}, then according to the properties of nonsingular \textit{M-matrices}\cite{plemmonsMmatrixCharacterizationsInonsingularMmatrices1977a}, 
there exists a positive vector $\mathbf{v}$ such that all elements in the vector $\mathbf{E}\mathbf{v}$ are positive. 

Hence, the coefficient vector can be constructed as
$
\mathbf{c}=(\boldsymbol{\sigma}^{\mathrm{T}}\mathbf{E}^{-1})^{\mathrm{T}}
$
with a positive vector $\boldsymbol{\sigma}\triangleq[\sigma_1,\dots,\sigma_n]^{\mathrm{T}}\geq 0$, and thus the inequality
$
\dot V\leq-\sum\nolimits_{i=1}^n\sigma_i\|x_i\|_2^2\leq 0
$ holds.
Therefore, the dissipation inequality in \eqref{equ_WholeSystem_dissipation} is satisfied, and the small-gain-like property holds.
\end{proof}

\begin{remark}
For the special case of a single-loop connection, the small-gain-like condition reduces to the criterion in \eqref{equ_singleloop_criterion}, as
\begin{equation}\label{equ_singleloop_criterion}
  a_{a1}>0,\quad a_{a2}>0,\quad\frac{b_{a12}}{a_{a1}}\frac{b_{a21}}{a_{a2}}<1.
\end{equation}
This condition is consistent with the loop-gain-based small-gain stability criterion for single-loop feedback systems \cite{khalilNonlinearSystemEdition2002}.
\end{remark}
}

In this way, the Lyapunov stability of the interconnected system can be characterized by the Lyapunov function in \eqref{equ_energy_like_fcn}, which satisfies the dissipation inequality in \eqref{equ_WholeSystem_dissipation} under the established small-gain-like property, thereby allowing large-signal stability evaluation of multi-converter systems. 

\vspace{-5pt}
\subsection{Dissipation Capability Analysis for Multi-converter Systems}\label{section_3b}
In this subsection, the dissipation capabilities of the selected $p$ GFL-GFM and $(q-1)$ GFM-GFM relative motions in Section~\ref{section_2c} are first analyzed and quantified. On this basis, a dissipation matrix is constructed for angle motions of the overall multi-converter system.

The equilibrium of the selected relative angle motions is first moved to zero by defining
\begin{equation}
  \Delta\delta_{ij}^{ll}=\delta_{ij}^{ll}-\delta_{ij}^{lle}\ 
  \Delta\delta_{ij}^{lm}=\delta_{ij}^{lm}-\delta_{ij}^{lme}\ 
  \Delta\delta_{ij}^{mm}=\delta_{ij}^{mm}-\delta_{ij}^{mme}
\end{equation}
where the superscript $e$ denotes the steady-state value of the corresponding relative angle. Accordingly, the selected GFL-GFM relative motions are mathematically described as

\begin{equation}
  \begin{aligned}\label{equ_RelativeLM_Active}
    &\Delta\dot\delta_{ij}^{lm}=A_{ij}^{lm}\Delta_{sin}(\Delta\delta_{ij}^{lm},\delta_{ij}^{lme}+\varphi_{ij}^{lm})\\
    &+\sum\nolimits_{k=1,k\neq i}^p\left[K_{pi}Z_{ik}^{ll}I_{lk}\Delta_{sin}(\Delta\delta_{ki}^{ll},\delta_{ki}^{lle}+\varphi_{lk}+\phi_{ik}^{ll})\right]\\
    &+\sum\nolimits_{r=1,r\neq j}^q\left[K_{pi}G_{ir}^{lm}U_{mr}\Delta_{sin}(\Delta\delta_{ir}^{lm},\delta_{ir}^{lme}-\phi_{ir}^{lm})\right]\\
    &+\sum\nolimits_{k=1,k\neq i}^{p}\left[\frac{G_{jk}^{ml}U_{mj}I_{lk}}{D_{pj}}\Delta_{cos}(\Delta\delta_{kj}^{lm},\delta_{kj}^{lme}+\varphi_{lk}+\phi_{jk}^{ml})\right]\\
    &+\sum\nolimits_{r=1,r\neq j}^q\left[\frac{Y_{jr}^{mm}U_{mj}U_{mr}}{D_{pj}}\Delta_{cos}(\Delta\delta_{rj}^{mm},\delta_{rj}^{mme}+\phi_{jr}^{mm})\right]
  \end{aligned}
\end{equation}
where $A_{ij}^{lm}$ is defined in \eqref{equ_Alm}, and $\varphi_{ij}^{lm}$ is determined with
\begin{equation}\label{equ_philm}
    \begin{aligned}
        &\sin{\varphi_{ij}^{lm}}=\frac{G^{lm}_{ij}U_{mj}\left(-K_{pi}\sin{\phi_{ij}^{lm}}+\frac{I_{li}}{D_{pj}}\cos{\left(\varphi_{li}+\phi_{ij}^{lm}\right)}\right)}{A_{ij}^{lm}}\\
        &\cos{\varphi_{ij}^{lm}}=\frac{G^{lm}_{ij}U_{mj}\left(K_{pi}\cos{\phi_{ij}^{lm}}-\frac{I_{li}}{D_{pj}}\sin{\left(\varphi_{li}+\phi_{ij}^{lm}\right)}\right)}{A_{ij}^{lm}}
  \end{aligned}
\end{equation}
and the functions $\Delta_{sin}(x,a)$ and $\Delta_{cos}(x,a)$ are defined as
\begin{equation}
  \begin{aligned}
  \Delta_{sin}(x,a)&=\sin{\left(x+a\right)}-\sin{a}\\
  \Delta_{cos}(x,a)&=\cos{\left(x+a\right)}-\cos{a}.
  \end{aligned}
\end{equation}

The selected GFM-GFM relative angle motions are
\begin{equation}\label{equ_RelativeMM_Active}
  \begin{aligned}
    &\Delta\dot\delta_{ij}^{mm}=A_{ij}^{mm}\Delta_{sin}(\Delta\delta_{ij}^{mm},\delta_{ij}^{mme}+\varphi_{ij}^{mm})\\
    &-\sum\nolimits_{k=1}^p\left[\frac{G_{ik}^{ml}U_{mi}I_{lk}}{D_{pi}}\Delta_{cos}(\Delta\delta_{ki}^{lm},\delta_{ki}^{lme}+\varphi_{lk}+\phi_{ik}^{ml})\right]\\
    &-\sum\nolimits_{r=1,r\neq i,j}^q\left[\frac{Y_{ir}^{mm}U_{mi}U_{mr}}{D_{pi}}\Delta_{cos}(\Delta\delta_{ri}^{mm},\delta_{ri}^{mme}+\phi_{ir}^{mm})\right]\\
    &+\sum\nolimits_{k=1}^p\left[\frac{G_{jk}^{ml}U_{mj}I_{lk}}{D_{pj}}\Delta_{cos}(\Delta\delta_{kj}^{lm},\delta_{kj}^{lme}+\varphi_{lk}+\phi_{jk}^{ml})\right]\\
    &+\sum\nolimits_{r=1,r\neq i,j}^q\left[\frac{Y_{jr}^{mm}U_{mj}U_{mr}}{D_{pj}}\Delta_{cos}(\Delta\delta_{rj}^{mm},\delta_{rj}^{mme}+\phi_{jr}^{mm})\right]
  \end{aligned}
\end{equation}
where $A_{ij}^{mm}$ is defined in \eqref{equ_Amm}, and $\varphi_{ij}^{mm}$ is determined with
\begin{equation}\label{equ_phimm}
  \begin{aligned}
        &\sin{\varphi_{ij}^{mm}}=\frac{Y_{ij}^{mm}U_{mi}U_{mj}\left(-1/D_{pi}+1/D_{pj}\right)\cos\phi_{ij}^{mm}}{A_{ij}^{mm}}\\
        &\cos{\varphi_{ij}^{mm}}=\frac{-U_{mi}U_{mj}Y_{ij}^{mm}\left(1/D_{pi}+1/D_{pj}\right)\sin\phi_{ij}^{mm}}{A_{ij}^{mm}}.
  \end{aligned}
\end{equation}

The selected $n_{\delta}$ ($n_{\delta}=p+q-1$) relative angles are further represented with a state vector $\mathbf{x}_{\delta}\triangleq [x_{\delta 1},\dots,x_{\delta n_{\delta}}]^{\mathrm{T}}$, where $x_{\delta i}\ (i=1,\dots,p)$ are the $p$ GFL-GFM relative angles $\Delta\delta_{ij}^{lm}$, and $x_{\delta i}\ (i=p+1,\dots,n_{\delta})$ are the $(q-1)$ GFM-GFM relative angles $\Delta\delta_{ij}^{mm}$. In this manner, the angle dynamics of the whole system are depicted by $n_{\delta}$ subsystems $\Sigma_{\delta i}$, as
\begin{equation}
  \Sigma_{\delta i}: \dot x_{\delta i}=f_{\delta i}(\mathbf{x}_{\delta}) \quad(i=1,\dots,n_{\delta})
\end{equation}
and $f_{\delta i}$ can be explicitly expressed through \eqref{equ_RelativeLM_Active} and \eqref{equ_RelativeMM_Active}.
With Hadamard's lemma, the relative angle dynamics in the multi-converter system can be further described as
\begin{equation}
  \Sigma_{\delta i}: \dot x_i=\sum\nolimits_{j=1}^{n_{\delta}}\left[g_{ixj}(\mathbf{x}_{\delta})x_{\delta j}\right]    \quad(i=1,\dots,n_{\delta})
\end{equation}
where the coefficient functions $g_{ixj}$ are calculated through
\begin{equation}\label{equ_hadamard}
  g_{ixj}(\mathbf{x}_{\delta})=\int_{0}^{1}\frac{\partial f_{\delta i}}{\partial x_{\delta j}}(t\mathbf{x}_{\delta})dt \quad (i,j=1,\dots,n_{\delta}).
\end{equation}

{\color{black}
Utilizing the positive definite and radially unbounded function $V_{\delta i}=\frac{1}{2}x^2_{\delta i}$ as the Lyapunov function for the subsystem $\Sigma_{\delta i}$. For a prescribed relative-angle domain $\Omega_\delta$, the dissipation inequality of the subsystem can be described as
\begin{equation}\label{equ_disp_ineq_relative}
  \begin{aligned}
    \dot V_{\delta i}&=x_{\delta i}f_{\delta i}(\mathbf{x}_{\delta})=g_{ixi}(\mathbf{x}_{\delta})x_{\delta i}^2+\sum\nolimits_{j=1,j\neq i}^{n_{\delta}}\left[g_{ixj}(\mathbf{x}_{\delta})x_{\delta i}x_{\delta j}\right]\\
    &\leq \max_{\mathbf{x}_{\delta}\in\Omega_{\delta}}{\left(g_{ixi}(\mathbf{x}_{\delta})\right)}x_{\delta i}^2\\
    &\quad +\sum\nolimits_{j=1,j\neq i}^{n_{\delta}}\max_{\mathbf{x}_{\delta}\in\Omega_{\delta}}{\left|g_{ixj}(\mathbf{x}_{\delta})\right|}\left(\frac{x_{\delta i}^2}  {2\epsilon_{ij}^2}+\frac{\epsilon_{ij}^2x_{\delta j}^2}{2}\right)\\
    &\triangleq -a_{\delta i}x_{\delta i}^2+\sum\nolimits_{j=1,j\neq i}^{n_{\delta}} b_{\delta ij}x_{\delta j}^2
  \end{aligned}
\end{equation}
where the inequality is derived through Young's inequality and $\epsilon_{ij}>0$ are the corresponding Young's coefficients.
With the obtained dissipation inequalities of subsystems, the dissipation matrix for the multi-converter system $\Sigma_{\delta}$ is built, and it is represented as
\begin{align}\label{equ_Ematrix_Forsystem}
    &\mathbf{E}_{\delta}=\diag(a_{\delta 1},\dots,a_{\delta n_{\delta}})-\mathbf{B}_{\delta}, \  
    \left(\mathbf{B}_{\delta}\right)_{ij}=
    \left\{
        \begin{aligned}
          &0\ \ \ \ i=j\\[-2pt]
          &b_{\delta ij}\  i\neq j
        \end{aligned}
    \right.\\\nonumber
    &\left\{
      \begin{aligned}
        a_{\delta i}&=-\max_{\mathbf{x}_{\delta}\in\Omega_{\delta}}{\left(g_{ixi}(\mathbf{x}_{\delta})\right)}-\sum\nolimits_{j=1,j\neq i}^{n_{\delta}}\frac{\max_{\mathbf{x}_{\delta}\in\Omega_{\delta}}{\left|g_{ixj}(\mathbf{x}_{\delta})\right|}}{2\epsilon_{ij}^2}\\
        b_{\delta ij}&=\max_{\mathbf{x}_{\delta}\in\Omega_{\delta}}{\left|g_{ixj}(\mathbf{x}_{\delta})\right|}\epsilon_{ij}^2/2
      \end{aligned}
    \right. .
\end{align}
}

\vspace{-5pt}
\subsection{Large-signal Stability Evaluation with Small-gain-like Property}\label{section_3c}
With the dissipation matrix, the small-gain-like property of the multi-converter system is checked and further established within certain angle limits in this subsection. Based on this, a Lyapunov function $V_{\delta}$ for multi-converter systems is constructed. Furthermore, an ellipsoidal stability region is derived, which enables direct large-signal stability evaluation for multi-converter systems.

{\color{black}
As elaborated and proved in Theorem~\ref{thm_smallgain}, the small-gain-like property holds if the corresponding dissipation matrix is a nonsingular \textit{M-matrix}. 
For the dissipation matrix $\mathbf{E}_{\delta}$ of the angle dynamics in multi-converter systems, as shown in \eqref{equ_Ematrix_Forsystem}, the elements in $\mathbf{E}_{\delta}$ are maximum-value functions of the relative angle variables $\mathbf{x}_{\delta}$. As a result, the dissipation matrix varies along the transient angle trajectories, and the small-gain-like property of system $\Sigma_{\delta}$ cannot be directly verified over the entire state space. 
Therefore, to explicitly characterize the small-gain-like property of system $\Sigma_{\delta}$, the transient angle variables are assumed to remain within an angle-limit region $\mathcal{B}_x$ during transients, where
\begin{equation} 
  \mathcal{B}_{x}(\mathbf{x}_{\delta}^{lim}) \triangleq \left\{ \mathbf{x}_{\delta} \,\middle|\, -\mathbf{x}_{\delta}^{lim} \leq \mathbf{x}_{\delta} \leq \mathbf{x}_{\delta}^{lim} \right\}
\end{equation}
with $\mathbf{x}_{\delta}^{lim}\triangleq [x_{\delta1}^{lim},\dots,x_{\delta n_{\delta}}^{lim}]^\mathrm{T}>0$. 
Accordingly, for $\mathbf{x}_{\delta}\in \mathcal{B}_{x}$, 
the maximum values in \eqref{equ_Ematrix_Forsystem} can be explicitly calculated as
\begin{equation}
  \bar g_{ixi}\triangleq\max_{\mathbf{x}_{\delta}\in\mathcal{B}_x}\left(g_{ixi}(\mathbf{x}_{\delta})\right)\quad \bar g_{ixj}\triangleq\max_{\mathbf{x}_{\delta}\in\mathcal{B}_x}\left|g_{ixj}(\mathbf{x}_{\delta})\right|  
\end{equation}
and thus the elements in $\mathbf{E}_{\delta}$ are rewritten as
\begin{equation}
  a_{\delta i}=-\bar g_{ixi}-\sum_{j=1,j\neq i}^{n_{\delta}}\frac{\bar g_{ixj}}{2\epsilon_{ij}^2}\quad
b_{\delta ij}=\frac{\epsilon_{ij}^2}{2}\bar g_{ixj}.
\end{equation}
In this manner, the dissipation matrix $\mathbf{E}_{\delta}$ can be rewritten in a constant form within the prescribed angle-limit region $\mathcal{B}_x(\mathbf{x}_{\delta}^{lim})$. 
Based on this constant-form matrix, the small-gain-like property of the multi-converter system can be directly examined throughout the region $\mathcal{B}_x(\mathbf{x}_{\delta}^{lim})$. 
}

Consequently, a group of admissible angle-limit sets can be identified, within which the multi-converter system holds the small-gain-like property. For a specific angle-limit set $\mathbf{x}_{\delta}^{lim}$, a Lyapunov function can be constructed within the associated region $\mathcal{B}_{x}(\mathbf{x}_{\delta}^{lim})$, as $V_{\delta}(\mathbf{x}_{\delta})=\mathbf{c}_{\delta}^{\mathrm{T}}\bar{V}_{\delta}$, where 
\begin{equation}
  \bar{V}_{\delta}\triangleq [V_{\delta 1},\dots,V_{\delta n_{\delta}}]^\mathrm{T}=\frac{1}{2}[x_{\delta 1}^2,\dots,x_{\delta n_{\delta}}^2]^{\mathrm{T}}
\end{equation}
and the coefficient vector $\mathbf{c}_{\delta}\triangleq[{c}_{\delta 1},\dots,{c}_{\delta n_{\delta}}]^{\mathrm{T}}$ is given by $\mathbf{c}_{\delta}=(\boldsymbol{\sigma}_{\delta}^\mathrm{T}\mathbf{E}_{\delta}^{-1})^\mathrm{T}$ with arbitrary $\boldsymbol{\sigma}_{\delta}\triangleq[\sigma_{\delta 1},\dots,\sigma_{\delta n_{\delta}}]^{\mathrm{T}}>0$. 
As a result, the establishment of the small-gain-like property ensures that the multi-converter system is Lyapunov stable within the corresponding angle-limit region $\mathcal{B}_{x}$. 
\begin{figure}\centering
    \includegraphics[scale=0.55]{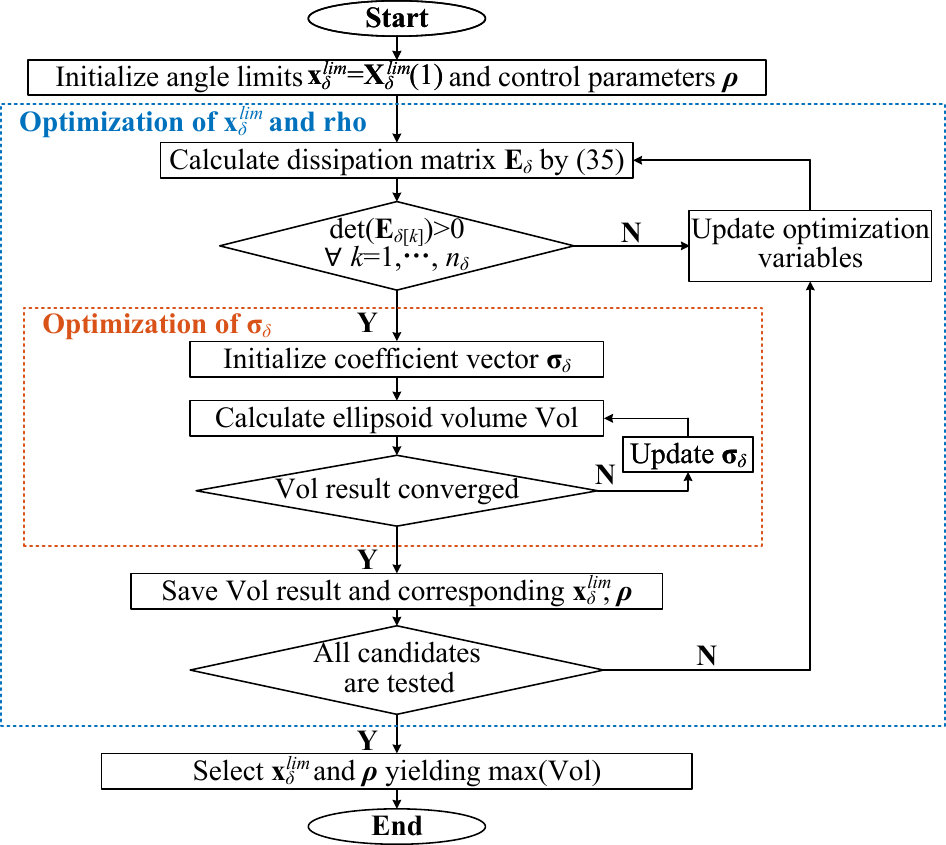}
    \vspace{-5pt}
    \caption{\textcolor{black}{Calculation process of $\mathbf{x}_{\delta}^{lim}$ for the maximum ellipsoid volume}}
    \label{Fig3_Flowchart}
    \vspace{-5pt}    
\end{figure}

{\color{black}
\begin{definition}[Lyapunov stability region]\label{Def_Bx}
Let $\mathbf{x}_{\delta}^{lim}\triangleq [x_{\delta1}^{lim},\dots,x_{\delta n_{\delta}}^{lim}]^\mathrm{T}>0$ be an admissible angle-limit vector such that the corresponding region $\mathcal{B}_{x}(\mathbf{x}_{\delta}^{lim})$ satisfies the small-gain-like property. The Lyapunov stability region is defined as the sublevel set of $V_{\delta}$ given by
\begin{equation}\label{equ_criterion_MM}
    \mathcal{B}_{V}(\mathbf{x}_{\delta}^{lim}) \triangleq \left\{ \mathbf{x}_{\delta} \,\middle|\, V_{\delta}(\mathbf{x}_{\delta}) \leq 
    \min_{\mathbf{x}_{\delta} \in \partial\mathcal{B}_{x}(\mathbf{x}_{\delta}^{lim})} V_{\delta}(\mathbf{x}_{\delta}) \right\}.
\end{equation}
\end{definition}
For notational convenience, define the boundary energy as 
\begin{equation}\label{equ_critical_engy}
    V_{\delta}^{cr}(\mathbf{x}_{\delta}^{lim}) \triangleq \min_{\mathbf{x}_{\delta} \in \partial\mathcal{B}_{x}(\mathbf{x}_{\delta}^{lim})} V_{\delta}(\mathbf{x}_{\delta}).
\end{equation}
\begin{proposition}[Forward invariance]
  For any admissible $\mathbf{x}_{\delta}^{lim}$ as defined in Definition~\ref{Def_Bx}, the set 
  $\mathcal{B}_{V}(\mathbf{x}_{\delta}^{lim})$ is forward-invariant.
\end{proposition}
\begin{proof}
  Since the small-gain-like property holds within 
  $\mathcal{B}_{x}(\mathbf{x}_{\delta}^{lim})$ and $\mathcal{B}_{V}(\mathbf{x}_{\delta}^{lim}) \subseteq \mathcal{B}_{x}(\mathbf{x}_{\delta}^{lim})$, the dissipation inequality 
  $\dot V_{\delta} \le 0$ holds throughout $\mathcal{B}_{V}(\mathbf{x}_{\delta}^{lim})$. 
  Therefore, $V_{\delta}$ is non-increasing along system trajectories. 
  For any initial condition in $\mathcal{B}_{V}(\mathbf{x}_{\delta}^{lim})$, the trajectory cannot cross the level set defined by $V_{\delta}^{cr}$, and thus remains within $\mathcal{B}_{V}(\mathbf{x}_{\delta}^{lim})$ for all future time.
\end{proof}

\begin{remark}
The constructed sum-type Lyapunov function exhibits a certain decay feature over time. According to \eqref{equ_WholeSystem_dissipation_v2}, the derivative of $V_\delta$ satisfies
\begin{equation}
\dot V_{\delta}\leq -\sum_{i=1}^{n_{\delta}}{\sigma_{\delta i}x_{\delta i}^2}
\leq -\min\limits_i\left(\frac{2\sigma_{\delta i}}{c_{\delta i}}\right)V_{\delta}
\triangleq -\lambda V_{\delta}.
\end{equation}
Thus, the decay performance of the sum-type Lyapunov function can be upper-bounded as
\begin{equation}
V_{\delta}(t)\leq V_{\delta}(0)e^{-\lambda t}.
\end{equation}
Therefore, the overall dissipation trend of the post-fault trajectory can be quantitatively characterized, which may provide insight into the convergence tendency and relative stability-recovery speed.
\end{remark}

}


Consequently, the $n_{\delta}$-dimensional ellipsoidal region $\mathcal{B}_{V}$ provides an effective estimate of the large-signal stability region for the multi-converter system.
{\color{black}
Moreover, based on the standard volume formula of an $n$-dimensional ellipsoid
\cite{boydConvexOptimization2004}, the volume of the Lyapunov stability region $\mathcal{B}_{V}$ is given by
\begin{equation}\label{equ_EllipVolume}
  \mathrm{Vol}(\mathcal{B}_V)
  = \frac{\pi^{n_\delta/2}}{\Gamma(n_\delta/2 + 1)} (2V_{\delta}^{\mathrm{cr}})^{n_\delta/2}
  \prod_{i=1}^{n_\delta} c_{\delta i}^{-1/2}
\end{equation}
where $\Gamma(\cdot)$ is the Gamma function, defined as
\begin{align}
\Gamma\left(\frac{n}{2}+1\right) =
\begin{cases}
\left(n/2\right)!, & n \in 2\mathbb{Z}, \\
\dfrac{n!!}{2^{\frac{n+1}{2}}}\sqrt{\pi}, & n \in 2\mathbb{Z} + 1,
\end{cases}
\end{align}
and $n!$ and $n!!$ denote the factorial and double factorial, respectively.
}

{\color{black}
It is worth noting that the obtained ellipsoidal Lyapunov stability region serves as a conservative estimation of the actual stability region. On the one hand, the interaction terms among relative angle motions are uniformly treated as destabilizing disturbances in the interconnected-system-based stability characterization in \eqref{equ_Subsystem_dissipation}, whereas some favorable interactions may contribute to synchronization restoration under specific transient conditions. 
On the other hand, to obtain an explicit and solvable dissipation inequality form, maximum-value approximation and Young's inequality are employed in \eqref{equ_disp_ineq_relative}, which leads to a conservative description of the actual dissipation capabilities of relative angle motions.
}

With the explicit expression in \eqref{equ_EllipVolume}, a relatively large stability region estimate can be obtained by evaluating the small-gain-like property over multiple candidate angle-limit sets $\mathbf{X}_{\delta}^{lim}\triangleq[\mathbf{x}_{\delta 1}^{lim},\dots,\mathbf{x}_{\delta m}^{lim}]$ through the automated procedure illustrated in \figref{Fig3_Flowchart}. For each candidate $\mathbf{x}_{\delta i}^{lim}=\mathbf{X}_{\delta}^{lim}(i)$, the vector $\boldsymbol{\sigma}_{\delta}$ is optimized via a nonlinear solver to maximize the corresponding ellipsoid volume. The angle-limit set that yields the largest ellipsoid volume is then selected for large-signal stability evaluation, and the resulting ellipsoid serves as an effective estimate of the actual stability region.
{\color{black}
Moreover, the control parameters can also be included as optimization variables, as indicated in \figref{Fig3_Flowchart} by the control parameter vector $\boldsymbol{\rho}$, through which coordinated control parameter tuning among converters can be achieved for large-signal stability enhancement of multi-converter systems. 
}

{\color{black}
Consequently, by establishing the small-gain-like property in multi-converter systems with respect to angle limits and system parameters, the proposed small-gain-like framework enables direct large-signal stability assessment for systems containing different types of converters. 
Specifically, the high-dimensional ellipsoid defined in \eqref{equ_criterion_MM} provides an explicit stability region in the corresponding state space, with which critical fault-clearing states can be identified under various transient conditions.
On this basis, together with the calculation procedure in \figref{Fig3_Flowchart}, the proposed framework provides a potential basis for quantitative stability region estimation and offers a systematic perspective for stability-oriented parameter design in multi-converter systems.
}
\section{Case Study}\label{section_4}
Based on the proposed stability evaluation method, the large-signal stability of a paralleled GFL-GFM grid-connected system and a two-area four-converter system is thoroughly studied in this section. The small-gain-like properties of both systems are assessed and Lyapunov functions are constructed correspondingly. In this way, the ellipsoidal stability regions can be determined for both systems, which provide effective approximations of their actual large-signal stability regions.
\begin{figure}\centering
    \includegraphics[scale=0.95]{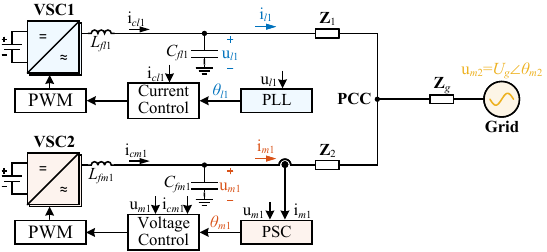}
    \vspace{-5pt}
    \caption{Structure of the paralleled GFL-GFM system}
    \label{Fig4_GFLGFMwithEg}
    \vspace{-5pt}    
\end{figure}

\begin{table}[tb]
    \vspace{-3pt}
    \centering
    \caption{Parameters of the Paralleled System}
    \vspace{-5pt}
    \renewcommand\arraystretch{1.2}
    \begin{tabular}{c|c|c}
    \hline
    Parameters & Descriptions & Values\\    \hline
    $I_{l1}e^{j\varphi_{l1}}$ & Current reference of VSC1  & 1 p.u.\\
    $U_{m1}$ & Voltage reference of VSC2 & 1 p.u.\\
    $P_{0m1}$ & Power reference of VSC2  & -0.5 p.u.\\
    $U_{g}$ & Grid voltage amplitude & 1 p.u.\\
    $K_{p1}$ & PLL proportional gain & 0.2 p.u.\\
    $D_{p1}$ & Droop coefficient of VSC2 & 20 p.u.\\ 
    $\mathbf{Z}_1$, $\mathbf{Z}_2$ & impedances connected to the converters & j0.1, j0.05 p.u.\\
    $\mathbf{Z}_g$ & Grid-connected impedance & j1.0 p.u.\\
    \hline
    \end{tabular}
    \label{tabel_parameters}
\end{table}

\subsection{Paralleled GFL-GFM Grid-connected System}\label{section_4a}
The paralleled system with two converters is first studied, as illustrated in \figref{Fig4_GFLGFMwithEg}. VSC1 and VSC2 operate as GFL and GFM converters respectively. They are connected to the point of coupling (PCC) through the circuit impedances $\mathbf{Z}_1$ and $\mathbf{Z}_2$, and deliver power to the grid via $\mathbf{Z}_g$. Define VSC1, VSC2 and the grid as nodes 1 to 3, the network admittance matrix is
\begin{equation}
  \mathbf{Y}_{net}^a=\frac{1}{\mathbf{ZZ}_{\Sigma}}
  \begin{bmatrix}
  \mathbf{Z}_2+\mathbf{Z}_g & -\mathbf{Z}_g & -\mathbf{Z}_2\\[-3pt]
  -\mathbf{Z}_g & \mathbf{Z}_1+\mathbf{Z}_g & -\mathbf{Z}_1\\[-3pt]
  -\mathbf{Z}_2 & -\mathbf{Z}_1 & \mathbf{Z}_1+\mathbf{Z}_2\\[-3pt]
    \end{bmatrix}
\end{equation}
where $\mathbf{ZZ}_{\Sigma}=\mathbf{Z}_1\mathbf{Z}_2+\mathbf{Z}_1\mathbf{Z}_g+\mathbf{Z}_2\mathbf{Z}_g$. 
Since the infinite-bus grid maintains a constant voltage amplitude and frequency, it can be modeled as a special case of GFM converter with voltage amplitude $U_{m2}=U_g$ and an infinite droop coefficient $D_{p2}\to \infty$. Through the self-damping evaluation introduced in Section~\ref{section_2c}, the relative angle motion between VSC1 and VSC2, denoted as $\Delta\delta_{11}^{lm}$, is selected as the GFL-GFM relative motion, and the relative angle motion between VSC2 and the grid, denoted as $\Delta\delta_{12}^{mm}$, is selected as the GFM-GFM relative motion. Define the state vector of the paralleled system as $\mathbf{x}_a=[\Delta\delta_{11}^{lm},\Delta\delta_{12}^{mm}]^{\mathrm{T}}\triangleq[x_{a1},x_{a2}]^{\mathrm{T}}$, the relative angle dynamics are mathematically depicted as
\begin{equation}
  \begin{aligned}\label{equ_Parallel_1}
    \dot x_{a1}&=A_{11}^{lm}\Delta_{sin}(x_{a1},\delta_{11}^{lme}+\varphi_{11}^{lm})\\
    &+K_{p1}G_{12}^{lm}U_{g}\Delta_{sin}(x_{a1}+x_{a2},\delta_{12}^{lme}-\phi_{12}^{lm})\\
    &+\frac{Y_{12}^{mm}U_{m1}U_{g}}{D_{p1}}\Delta_{cos}(-x_{a2},\delta_{21}^{mme}+\phi_{12}^{mm})\\
    &\triangleq g_{1x1}^a(\mathbf{x}_{a})x_{a1}+g_{1x2}^a(\mathbf{x}_{a})x_{a2}\\
    \dot x_{a2}&=A_{12}^{mm}\Delta_{sin}(x_{a2},\delta_{12}^{mme}+\varphi_{12}^{mm})\\
    &-\frac{G_{11}^{ml}U_{m1}I_{l1}}{D_{p1}}\Delta_{cos}(x_{a1},\delta_{11}^{lme}+\varphi_{l1}+\phi_{11}^{ml})\\
    &\triangleq g_{2x2}^a(\mathbf{x}_{a})x_{a2}+g_{2x1}^a(\mathbf{x}_{a})x_{a1}
  \end{aligned}
\end{equation}

\begin{figure}\centering
   \includegraphics[scale=1.3]{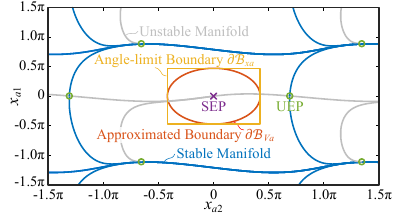}\\
    \vspace{-3pt}
    \quad\ \  \footnotesize (a)\\
    \includegraphics[scale=1.3]{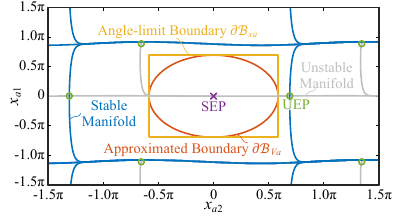}\\
    \vspace{-3pt}
    \quad\ \ \  \footnotesize (b)\\
    \vspace{-5pt}
    \caption{Boundary approximation results for paralleled systems. (a) Original parameters. (b) Optimized parameters.}
    \label{Fig5_Paralleled_Phaseplane}
\end{figure}

The specific expressions of the amplitudes $A_{11}^{lm}$, $A_{12}^{mm}$ and the phases $\varphi_{11}^{lm}$, $\varphi_{12}^{mm}$ are solved with \eqref{equ_Alm}, \eqref{equ_Amm}, \eqref{equ_philm} and \eqref{equ_phimm}.
The coefficient functions $g_{ixk}^a$ are further calculated with the Hadamard's lemma in \eqref{equ_hadamard}, and their detailed expressions are
\begin{equation}
  \begin{aligned}
  g_{1x1}^a(\mathbf{x}_{a})&=A_{11}^{lm}{\mathcal{F}_{sin}(x_{a1},\delta_{11}^{lme}+\varphi_{11}^{lm})}\\
  &+K_{p1}G_{12}^{lm}U_{g}\mathcal{F}_{sin}(x_{a1}+x_{a2},\delta_{12}^{lme}-\phi_{12}^{lm})\\
  g_{1x2}^a(\mathbf{x}_{a})&=K_{p1}G_{12}^{lm}U_{g}\mathcal{F}_{sin}(x_{a1}+x_{a2},\delta_{12}^{lme}-\phi_{12}^{lm})\\
  &+\frac{Y_{12}^{mm}U_{m1}U_{g}}{D_{p1}}\mathcal{F}_{cos}(-x_{a2},\delta_{21}^{mme}+\phi_{12}^{mm})\\
  g_{2x1}^a(\mathbf{x}_{a})&=-\frac{G_{11}^{ml}U_{m1}I_{l1}}{D_{p1}}\mathcal{F}_{cos}(x_{a1},\delta_{11}^{lme}+\varphi_{l1}+\phi_{11}^{ml})\\
  g_{2x2}^a(\mathbf{x}_{a})&=A_{12}^{mm}\mathcal{F}_{sin}(x_{a2},\delta_{12}^{mme}+\varphi_{12}^{mm})
  \end{aligned}
\end{equation}
where $\mathcal{F}_{sin}\triangleq\frac{\Delta_{sin}(x,a)}{x}$ and $\mathcal{F}_{cos}\triangleq\frac{\Delta_{cos}(x,a)}{x}$, respectively. Accordingly, the entries in the dissipation matrix $\mathbf{E}_a$ can be directly solved using \eqref{equ_Ematrix_Forsystem}, and the small-gain-like property of the paralleled system holds when $a_{a1}$ and $a_{a2}$ are positive and  $\frac{b_{a12}}{a_{a1}}\frac{b_{a21}}{a_{a2}}<1$.
The ratios $\frac{b_{a12}}{a_{a1}}$ and $\frac{b_{a21}}{a_{a2}}$ are further expressed with the maximum values of coefficient functions, as 
\begin{equation}
  \begin{aligned}
    &\frac{b_{a12}}{a_{a1}}=\max_{\mathbf{x}_a}{\left|g_{1x2}^a\right|}\left[-\frac{2\max_{\mathbf{x}_a}{\left(g_{1x1}^a\right)}}{\epsilon_{12}^2}-\frac{\max_{\mathbf{x}_a}{\left|g_{1x2}^a\right|}}{\epsilon_{12}^4}\right]^{-1}\\
    &\frac{b_{a21}}{a_{a2}}=\max_{\mathbf{x}_a}{\left|g_{2x1}^a\right|}\left[-\frac{2\max_{\mathbf{x}_a}{\left(g_{2x2}^a\right)}}{\epsilon_{21}^2}-\frac{\max_{\mathbf{x}_a}{\left|g_{2x1}^a\right|}}{\epsilon_{21}^4}\right]^{-1}.
  \end{aligned}
\end{equation}
Note that $\max{\left(g_{ixi}^a\right)}$ are negative because $a_{ai}>0$. In this manner, the Young's coefficients are selected as
\begin{equation}
  \frac{1}{\epsilon_{12}^2}=\frac{-\max_{\mathbf{x}_a}{\left(g_{1x1}^a\right)}}{\max_{\mathbf{x}_a}{\left|g_{1x2}^a\right|}}\ 
  \frac{1}{\epsilon_{21}^2}=\frac{-\max_{\mathbf{x}_a}{\left(g_{2x2}^a\right)}}{\max_{\mathbf{x}_a}{\left|g_{2x1}^a\right|}}.
\end{equation}
Consequently, the small-gain-like property of the paralleled system in \figref{Fig4_GFLGFMwithEg} holds when
\begin{equation}\label{equ_IOS_Paralleled}
  \begin{aligned}
    &\max_{\mathbf{x}_a}{\left(g_{1x1}^a(\mathbf{x}_{a})\right)}<0\quad \max_{\mathbf{x}_a}{\left(g_{2x2}^a(\mathbf{x}_{a})\right)}<0\\
    &\left(\frac{\max_{\mathbf{x}_a}{\left|g_{1x2}^a(\mathbf{x}_{a})\right|}}{\max_{\mathbf{x}_a}{\left(g_{1x1}^a(\mathbf{x}_{a})\right)}}\right)^2\left(\frac{\max_{\mathbf{x}_a}{\left|g_{2x1}^a(\mathbf{x}_{a})\right|}}{\max_{\mathbf{x}_a}{\left(g_{2x2}^a(\mathbf{x}_{a})\right)}}\right)^2<1.
  \end{aligned}
\end{equation}

As a result, for a given relative angle-limit set $\mathbf{x}_{a}^{lim}\triangleq[x_{a1}^{lim},x_{a2}^{lim}]^\mathrm{T}$, if the inequalities in \eqref{equ_IOS_Paralleled} are satisfied for all $\mathbf{x}_a$ within
\begin{equation} 
  \mathcal{B}_{xa}(\mathbf{x}_{a}^{lim}) \triangleq \left\{ \mathbf{x}_{a} \,\middle|\, -\mathbf{x}_{a}^{lim} \leq \mathbf{x}_{a} \leq \mathbf{x}_{a}^{lim} \right\}
\end{equation}
then the small-gain-like property of the paralleled system is established within the corresponding angle limit, and a Lyapunov function can be constructed as
\begin{equation}
  V_a(\mathbf{x}_a)=\mathbf{c}_{a}^{\mathrm{T}}\bar{V}_{a}=\frac{1}{2}c_{a1}x_{a1}^2+\frac{1}{2}c_{a2}x_{a2}^2.
\end{equation}

The two-dimensional forward-invariant ellipsoidal stability region $\mathcal{B}_{Va}$ can be further calculated from \eqref{equ_critical_engy} and \eqref{equ_criterion_MM}.
Based on this, multiple angle-limit candidates are examined with the algorithm in \figref{Fig3_Flowchart}, and the one yielding the largest ellipse area is selected for large-signal stability evaluation. 

To verify the effectiveness of the proposed large-signal stability evaluation method, a paralleled GFL-GFM system with parameters listed in TABLE~\ref{tabel_parameters} is studied. Using the algorithm in \figref{Fig3_Flowchart}, an optimized angle-limit set $\mathbf{x}_{a}^{lim}=[1.485,1.322]^\mathrm{T}$ is obtained. The corresponding approximated boundary $\partial\mathcal{B}_{Va}$ is then visualized in the $x_{a1}-x_{a2}$ plane, as shown in \figref{Fig5_Paralleled_Phaseplane}(a). In the figure, the obtained ellipsoidal stability boundary $\partial\mathcal{B}_{Va}$ is represented by the orange line and lies entirely within the corresponding angle-limit boundary $\partial\mathcal{B}_{xa}$ shown in yellow. 
\textcolor{black}{
For second-order nonlinear systems, the stability boundary with respect to a stable equilibrium point (SEP) is represented by the one-dimensional stable manifolds of its adjacent unstable equilibrium points (UEPs)\cite{khalilNonlinearSystemEdition2002}.
Hence, the actual large-signal stability boundary of the paralleled system can be numerically traced through backward time-domain numerical integration and is visualized by the blue lines in \figref{Fig5_Paralleled_Phaseplane}.
}
It can be observed that the computed ellipse area lies within the actual boundary, which proves that the proposed method could provide effective stability assessment for the paralleled system.

Moreover, the control parameters $K_{p1}$ and $D_{p1}$ in the paralleled GFL-GFM system are further tuned to expand the stability region. The optimization of control parameters can be performed using the algorithm in \figref{Fig3_Flowchart}, where both the control parameters and angle-limit sets are considered as optimization variables. Through this process, the maximum ellipsoid volume, which corresponds to the largest stability region area, is searched within the predefined bounds.
For the original paralleled system in \figref{Fig5_Paralleled_Phaseplane}(a), an optimized control parameter group is obtained as $[K_{p1};D_{p1}]=[0.5;50]$, and the corresponding angle-limit set is $\mathbf{x}_{a}^{lim}=[2.182,1.840]^\mathrm{T}$. The resulting elliptical boundary $\partial\mathcal{B}_{Va}$ is illustrated by the orange line in \figref{Fig5_Paralleled_Phaseplane}(b), along with the angle-limit boundary $\partial\mathcal{B}_{xa}$. It is evident that the stability region $\mathcal{B}_{Va}$ is significantly enlarged with the optimized control parameters, which further validates the effectiveness of the proposed method for quantitative stability evaluation and control parameter tuning.
\begin{figure}\centering
    \includegraphics[scale=0.92]{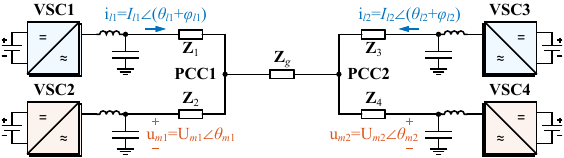}
    \vspace{-5pt}
    \caption{Structure of the two-area four-converter system}
    \label{Fig6_2A4M_Structure}
\end{figure}

\vspace{-5pt}
\subsection{Two-area Four-converter System}\label{section_4b}
The large-signal stability performance of a four-converter system is also analyzed via the proposed small-gain-like method. The system structure is shown in \figref{Fig6_2A4M_Structure}. Two converter groups are connected via a transmission line with the impedance $\mathbf{Z}_g$, where VSC1 and VSC3 operate as GFL converters and VSC2 and VSC4 serve as GFM converters, respectively. The impedances $\mathbf{Z}_1$ and $\mathbf{Z}_2$ connect VSC1 and VSC2 to PCC1, while the impedances $\mathbf{Z}_3$ and $\mathbf{Z}_4$ connect VSC3 and VSC4 to PCC2, separately. The currents of the GFL converters (VSC1 and VSC3) are defined as $ \mathrm{i}_{li}=I_{li}e^{j(\theta_{li}+\varphi_{li})},i=1,2$, where $\theta_{li}$ are the power angles generated by the PLLs. The voltages of the GFM converters (VSC2 and VSC4) are given by $\mathrm{u}_{mi}=U_{mi}e^{j\theta_{mi}},i=1,2$, where $\theta_{mi}$ are the power angles of GFM converters. Defining VSC1 and VSC3 as nodes 1 and 2, and VSC2 and VSC4 as nodes 3 and 4, the corresponding network admittance matrix $\mathbf{Y}_{net}^b$ can be calculated based on the circuit topology. Accordingly, there is
\begin{equation}
[\mathrm{i}_{l1},\mathrm{i}_{l2},\mathrm{i}_{m1},\mathrm{i}_{m2}]^\mathrm{T}=\mathbf{Y}_{net}^b[\mathrm{u}_{l1},\mathrm{u}_{l2},\mathrm{u}_{m1},\mathrm{u}_{m2}]^\mathrm{T}
\end{equation}

Through the comparison of self-damping terms derived from all relative motions between GFL and GFM converters, the relative motion groups (VSC1, VSC2), (VSC3, VSC4) and (VSC2, VSC4) are chosen for large-signal stability assessment. 
It is worth noting that the grouping results indicate that the nearest GFM converter of a GFL converter is the one that connected to the same PCC, which aligns with the network structure.
The state vector is defined as 
\begin{equation}
  \mathbf{x}_b=[\Delta\delta_{11}^{lm},\Delta\delta_{22}^{lm},\Delta\delta_{12}^{mm}]^{\mathrm{T}}\triangleq[x_{b1},x_{b2},x_{b3}]^{\mathrm{T}}
\end{equation}
and the angle dynamics of the four-converter system are mathematically depicted by three subsystems $\Sigma_{bi}$, where
\begin{equation}
  \Sigma_{bi}: \dot x_{bi}=f_{bi}(\mathbf{x}_{b})\triangleq \sum\nolimits_{j=1}^3{g_{ixj}^b(\mathbf{x}_{b})x_{bj}}\ (i,j=1,2,3)
\end{equation}

\begin{table}[tb]
    \vspace{-3pt}
    \centering
    \caption{Parameters of the Four-converter System}
    \vspace{-5pt}
    \renewcommand\arraystretch{1.2}
    \begin{tabular}{c|c|c}
    \hline
    Parameters & Descriptions & Values\\    \hline
    $I_{l1}e^{j\varphi_{l1}}$ & Current reference of VSC1  & 0.5 p.u.\\
    $I_{l2}e^{j\varphi_{l2}}$ & Current reference of VSC3  & -1 p.u.\\
    $U_{m1,2}$ & Voltage references of GFM converters & 1 p.u.\\
    $P_{0m1,2}$ & Power references of GFM converters & -0.2, 0.5 p.u.\\
    $K_{p1,2}$ & PLL proportional gains & 0.2 p.u.\\
    $D_{p1,2}$ & Droop coefficients of GFM converters& 20 p.u.\\ 
    $\mathbf{Z}_1$ & Impedance between VSC1 and PCC1 & 0.02+j0.2 p.u.\\
    $\mathbf{Z}_2$ & Impedance between VSC2 and PCC1 & 0.005+j0.05 p.u.\\
    $\mathbf{Z}_3$ & Impedance between VSC3 and PCC2 & 0.01+j0.1 p.u.\\
    $\mathbf{Z}_4$ & Impedance between VSC4 and PCC2 & 0.005+j0.05 p.u.\\
    $\mathbf{Z}_g$ & Impedance between PCC1 and PCC2 & 0.1+j1.0 p.u.\\
    \hline
    \end{tabular}
    \label{tabel_parameters_2A4M}
    \vspace{-10pt}
\end{table}
\begin{figure}\centering
    \includegraphics[scale=0.305]{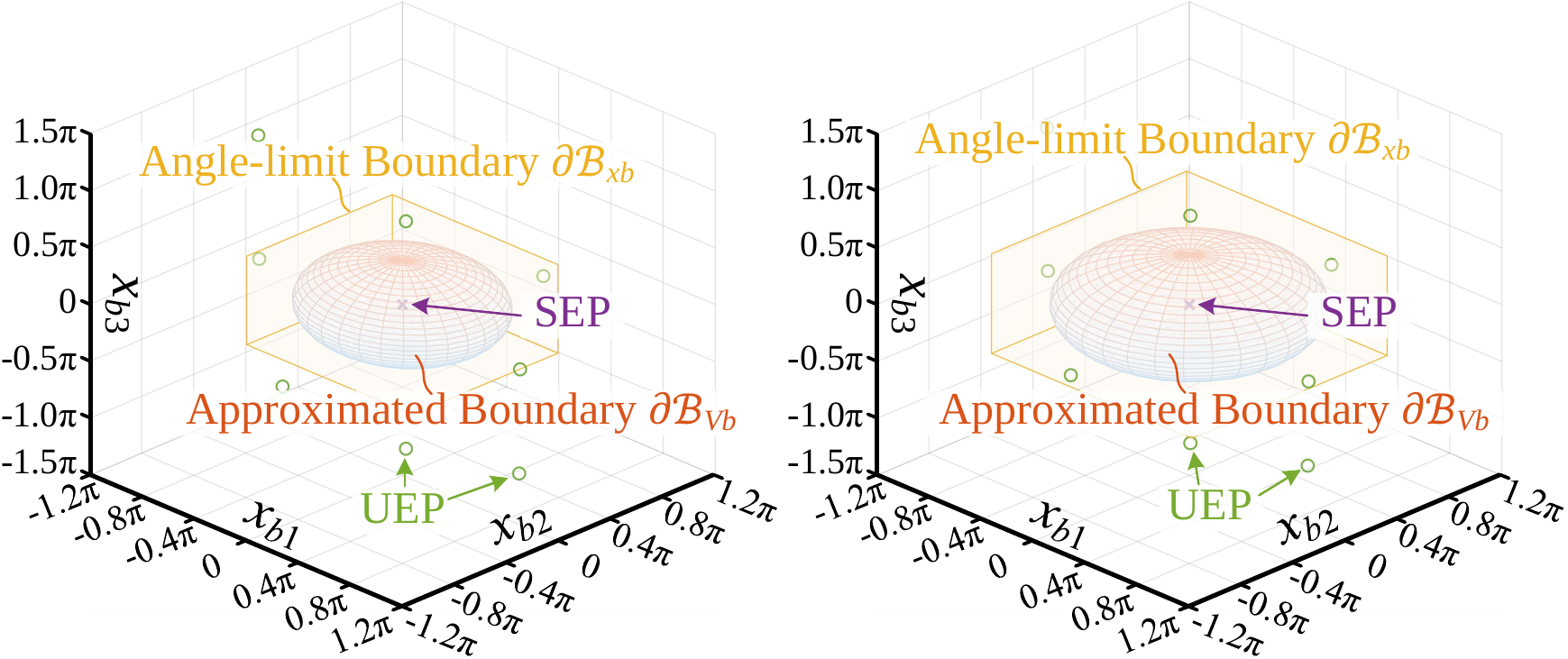}
    \vspace{-5pt}
    \ \footnotesize (a)\quad\quad\quad\quad\quad\quad\quad\quad\quad\quad\quad\quad\quad\quad\quad  \footnotesize (b)\\
    \caption{Boundary approximation results for four-converter systems. (a) Original parameters. (b) Optimized parameters.}
    \label{Fig7_2A4M_PhasePlane}
\end{figure}

The detailed expressions of $f_{bi}(\mathbf{x}_{b})$ are given in \eqref{equ_RelativeLM_Active} and \eqref{equ_RelativeMM_Active}, and the associated coefficient functions $g_{ixj}^b(\mathbf{x}_{b})$ are calculated using Hadamard's lemma in \eqref{equ_hadamard}. The dissipation matrix $\mathbf{E}_b$ for the four-converter system is then constructed as
\begin{align}\label{equ_Ematrix_2A4M}
    &\quad\quad\mathbf{E}_{b}=
    \begin{bmatrix}
    a_{b1}& -b_{b12}& -b_{b13}\\[-3pt]
    -b_{b21}& a_{b2}& -b_{b23}\\[-3pt]
    -b_{b31}& -b_{b32}& a_{b3}\\[-3pt]
    \end{bmatrix}\\\nonumber
    &\left\{
      \begin{aligned}
        a_{b i}&=-\max_{\mathbf{x}_b}{\left(g_{ixi}^b(\mathbf{x}_{b})\right)}-\sum\nolimits_{j=1,j\neq i}^{3}\max_{\mathbf{x}_b}{\left|g_{ixj}^b(\mathbf{x}_{b})\right|}/2\\[-2pt]
        b_{b ij}&=\max_{\mathbf{x}_b}{\left|g_{ixj}^b(\mathbf{x}_{b})\right|}/2
      \end{aligned}
    \right.
\end{align}
where all Young's coefficients $\epsilon_{ij}$ in \eqref{equ_Ematrix_Forsystem} are set to 1 for simplicity. In this manner, for the angles bounded by
\begin{equation} 
  \begin{aligned}
      \mathcal{B}_{xb}(\mathbf{x}_{b}^{lim})\triangleq \left\{ \mathbf{x}_{b} \,\middle|\, -\mathbf{x}_{b}^{lim} \leq \mathbf{x}_{b} \leq \mathbf{x}_{b}^{lim} \right\}\\
  \end{aligned}
\end{equation}
where $\mathbf{x}_{b}^{lim}\triangleq[x_{b1}^{lim},x_{b2}^{lim},x_{b3}^{lim}]^\mathrm{T}>0$ is the angle-limit set. If all leading principal minors of $\mathbf{E}_b$ are positive for any $\mathbf{x}_b$ in $\mathcal{B}_{xb}$, then the small-gain-like property of the four-converter system holds within the $\mathcal{B}_{xb}$. Accordingly, a Lyapunov function is built for the four-converter system, written as
\begin{equation}
  V_{b}(\mathbf{x}_b)=\mathbf{c}_b^\mathrm{T}\bar V_b=\frac{1}{2}\left(\boldsymbol{\sigma}_b^{\mathrm{T}}\mathbf{E}_b^{-1}\right)[ x_{b1}^2\ x_{b2}^2\ x_{b3}^2]^\mathrm{T}
\end{equation}
where the coefficient vector $\mathbf{c}_b$ is determined with $\boldsymbol{\sigma}_{b}$ and the dissipation matrix $\mathbf{E}_b$. According to the stability criterion in \eqref{equ_criterion_MM}, the large-signal stability region of the four-converter system can be estimated with a three-dimensional ellipsoid, as
\begin{equation}
    \mathcal{B}_{Vb} \triangleq \left\{ \mathbf{x}_{b} \,\middle|\,  V_{b}(\mathbf{x}_b)\leq \min_{\mathbf{x}_b \in \partial\mathcal{B}_{xb}(\mathbf{x}_{b}^{lim})} V_{b}(\mathbf{x}_{b})\right\}
\end{equation}

With the parameters in TABLE~\ref{tabel_parameters_2A4M}, the critical angle-limit set, $\mathbf{x}_b^{lim}=[1.994,1.753,1.221]^{\mathrm{T}}$, which yields the largest ellipsoid volume, is solved utilizing the algorithm in \figref{Fig3_Flowchart}. The corresponding approximated stability boundary is demonstrated in \figref{Fig7_2A4M_PhasePlane}(a) as the orange ellipsoid. In the same figure, a portion of the cubic angle-limit boundary $\partial B_{xb}$ is visualized as yellow planes, and all UEPs surrounding the SEP are marked as green circles. 
Based on the ellipsoidal region, the large-signal stability of the four-converter system can be directly and quantitatively evaluated by checking whether the post-fault relative angle lies in $\mathcal{B}_{Vb}$. In this manner, the computational burden for large-signal stability assessment is remarkably released.

Furthermore, the control parameters are optimized through the process in \figref{Fig3_Flowchart}, and an optimized parameter set $[K_{p1};K_{p2};D_{p1};D_{p2}]=[0.471;0.488;18.282;39.914]$ is obtained, with the corresponding angle-limit set $\mathbf{x}_b^{lim}=[2.415,2.355,1.375]^{\mathrm{T}}$. The estimated region of the optimized four-converter system is visualized in \figref{Fig7_2A4M_PhasePlane}(b), and it is evident that the volume of the ellipsoid $\mathcal{B}_{Vb}$ is greatly enlarged with the optimized parameters, which further demonstrates the capability of the proposed method in parameter optimization and stability evaluation.

\section{Experimental Validation}\label{section_5}
\begin{figure}\centering
    \includegraphics[scale=0.42]{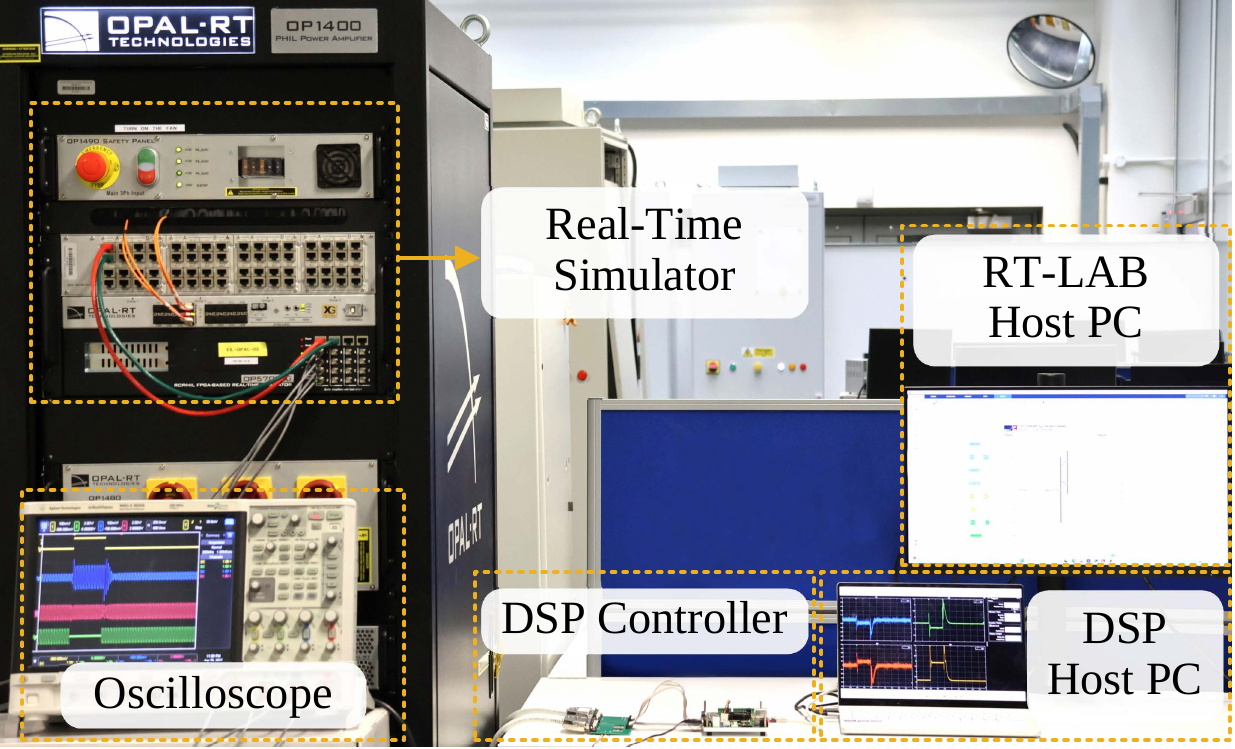}
    \vspace{-5pt}
    \caption{Control hardware-in-loop experimental platform}
    \label{Fig8_Platform}
\end{figure}

In this section, experimental validations are conducted using the RT-LAB based hardware-in-loop platform illustrated in \figref{Fig8_Platform}. The critical fault-clearing points under various fault types for both the paralleled and four-converter systems are identified using the proposed method and are further validated through time-domain transient waveforms.

The paralleled GFL-GFM system in Section~\ref{section_4a} is firstly examined. Consider a voltage sag fault occurring at the grid, the relative angles would diverge from the SEP during the fault and the system may become unstable if the fault duration exceeds a certain threshold. 
\textcolor{black}{The post-fault stability is thus characterized by the critical fault-clearing time (CCT), which is defined as the maximum fault duration that allows the system to converge back to the same SEP after fault clearance. By adopting the CCT as a key stability indicator, exhaustive numerical integration for actual boundary characterization in high-order systems can be effectively avoided, and the computational burden for identifying the actual critical stability condition is thereby significantly alleviated.}
 Moreover, the critical fault-clearing angle (CCA) is recognized as the angle $\delta^{mm}_{12}$ between the GFM converter and the grid at the fault-clearing moment when the fault-clearing time (FCT) equals the CCT. 
The CCT and CCA can be directly identified with the proposed method by locating the intersection point between the fault-on trajectory and the ellipsoidal boundary in \figref{Fig5_Paralleled_Phaseplane}, and the evaluation results for the paralleled systems with both the original and optimized control parameters are listed in TABLE~\ref{tabel_result_2M}, denoted by CCT$_A$ and CCA$_A$. For comparison, the actual CCT and CCA values, obtained through repetitive iteration process, are also given in the table and referred to as CCT$_R$ and CCA$_R$. It is clear that the estimated CCT$_A$ and CCA$_A$ results in both cases offer an effective approximation of the actual values, which demonstrates the validity of the proposed method.

Furthermore, the time-domain waveforms of the paralleled system with optimized parameters are illustrated in \figref{Fig9_Waveform_2M}(a) and \figref{Fig9_Waveform_2M}(b), with different FCTs of 729 ms and 882 ms, separately. 
{\color{black} For a voltage sag fault, the grid voltage drops to zero during the fault period, as indicated by the gray area in the figure. The grid voltage is restored at the fault-clearing moment marked by the light blue dashed line, which coincides with the end of the fault duration.
}
For the first case with FCT of 729 ms, the post-fault system is stable as the angles return to the original SEP, and the stability is correctly indicated by the proposed method as CCT$_A$ equals 729 ms. As FCT increases to 882 ms, the system becomes unstable as $\delta^{mm}_{12}$ undergoes a one-cycle oscillation and reaches a different SEP. 
{\color{black}It is noteworthy that the relative angles are folded into the interval $[-\pi,\pi)$ for clearer visualization, since electrical quantities in power systems are generally expressed as sums of trigonometric functions of relative angles, and these trigonometric functions are 2$\pi$-periodic. Nevertheless, the stable cases recognized by the proposed method would not exhibit cross-period transient relative angle trajectories, since the proposed method identifies a forward-invariant Lyapunov stability region associated with the original pre-fault SEP.}
The proposed evaluation method also predicts the instability as the CCT$_A$ is smaller than the FCT of 882 ms. As a result, the proposed method is proved to offer accurate stability assessments for the paralleled GFL-GFM system under voltage sag faults.
\begin{table}[tb]
  \centering
  \caption{Stability Evaluation Results of Parallel Systems}
  \vspace{-5pt}
  \renewcommand\arraystretch{1.55}
  \begin{tabular}{c|c|c|c|c|c|c}
    \hline
    & \makecell{$\mathrm{CCT}_A$ \\ (ms)} & \makecell{$\mathrm{CCA}_A$ \\ (rad)} 
    & \makecell{$\mathrm{CCT}_R$ \\ (ms)} & \makecell{$\mathrm{CCA}_R$ \\ (rad)} 
    & \makecell{$\varphi_{pj,A}^{\max}$ \\ (rad)} & \makecell{$\varphi_{pj,R}^{\max}$ \\ (rad)} \\
    \hline
    Original  & 207 & 1.81 & 338 & 2.65 & 1.32 & 2.17 \\
    Optimized & 729 & 2.32 & 860 & 2.66 & 1.84 & 2.17 \\
    \hline
  \end{tabular}
  \vspace{-6pt}
  \label{tabel_result_2M}
\end{table}
\begin{figure}\centering
   \includegraphics[scale=1.05]{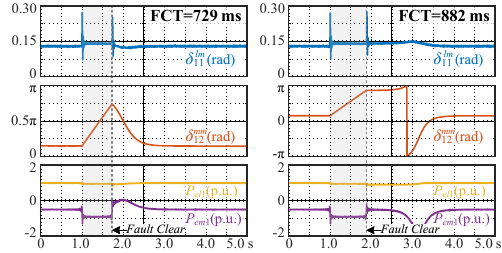}\\
    \vspace{-1pt}
    {\color{black}
    \hspace{0.8em}\ \footnotesize (a)\quad\quad\quad\quad\quad\quad\quad\quad\quad\quad\quad\quad\quad\quad\quad  \footnotesize (b)\\
    \vspace{1pt}}
    \includegraphics[scale=1.05]{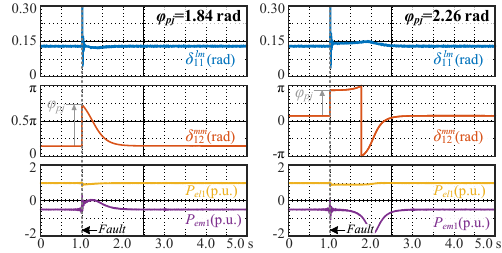}\\
    \vspace{-1pt}
    {\color{black}
    \hspace{0.8em}\ \footnotesize (c)\quad\quad\quad\quad\quad\quad\quad\quad\quad\quad\quad\quad\quad\quad\quad  \footnotesize (d)\\
    \vspace{-5pt}}
    \caption{Experimental results of the paralleled GFL-GFM system. (a) Voltage sag fault with FCT=729 ms. (b) Voltage sag fault with FCT=882 ms. (c) Phase jump fault with $\varphi_{pj}$=1.84 rad. (d) Phase jump fault with $\varphi_{pj}$=2.26 rad.}
    \label{Fig9_Waveform_2M}
\end{figure}

Moreover, the grid phase jump fault is also analyzed for the paralleled systems, where an abrupt phase shift occurs at the grid and may cause converters to lose synchronization with the grid when the jump angle is too large. The maximum allowable phase jump angle, referred to as $\varphi_{pj}^{max}$, is of key concern and the actual values for both paralleled systems are calculated through the iteration process and listed as $\varphi_{pj,R}^{max}$ in TABLE~\ref{tabel_result_2M}. Additionally, since the relative angle between converters, $\delta^{lm}_{11}$, remains unchanged while the angle between the GFM converter and the grid, $\Delta\delta^{mm}_{12}$, jumps instantaneously at the fault moment, the approximated maximum phase jump angle, denoted as $\varphi_{pj,A}^{max}$, can be obtained by identifying the intersection point between $\Delta\delta^{lm}_{11}=0$ and the ellipsoidal boundary in \figref{Fig5_Paralleled_Phaseplane}. By comparing the maximum phase jump angle results in TABLE~\ref{tabel_result_2M}, it is evident that $\varphi_{pj,A}^{max}$ provides an effective estimate of the actual $\varphi_{pj,R}^{max}$ for both cases.

The transient waveforms under the phase jump fault for the paralleled system with optimized parameters are shown in \figref{Fig9_Waveform_2M}(c) and \figref{Fig9_Waveform_2M}(d), where different phase jump angles $\varphi_{pj}$ are investigated. At t=1s, the grid undergoes an instantaneous phase shift, causing the angle $\delta_{12}^{mm}$ to jump simultaneously, as the gray arrow indicates. When the angle shift is 1.84 rad, the system returns to its original SEP after fault clearance, and the stability is correctly reflected by the estimated result, as $\varphi_{pj,A}^{max}$ of 1.84 rad equals the current phase jump angle. For the case with angle shift of 2.26 rad in \figref{Fig9_Waveform_2M}(d), the converters lose synchronization with the grid during the transient and eventually settle at a new SEP. The instability is also accurately indicated by the proposed method since the actual phase jump exceeds $\varphi_{pj,A}^{max}$. In this manner, the effectiveness of the proposed stability evaluation method for paralleled systems under phase jump faults is further validated.
\begin{table}[tb]
  \centering
  \caption{Stability Evaluation Results of Four-converter Systems}
  \vspace{-5pt}
  \renewcommand\arraystretch{1.3}
  \begin{tabular}{c|c|c|c|c}
    \hline
    & \multicolumn{2}{c|}{Short-circuit fault at $\mathbf{Z}_g$} & \multicolumn{2}{c}{Short-circuit fault at $\mathbf{Z}_1$} \\
    \cline{2-5}
    & \makecell{$\mathrm{CCT}_A$ (ms)} & \makecell{$\mathrm{CCT}_R$ (ms)}
    & \makecell{$\mathrm{CCT}_A$ (ms)} & \makecell{$\mathrm{CCT}_R$ (ms)}\\
    \hline
    Original  & 114 & 208 & 40 & 153 \\
    Optimized & 228 & 404 & 40 & 129\\
    \hline
  \end{tabular}
  \vspace{-6pt}
  \label{tabel_result_4M}
\end{table}
\begin{figure}\centering
   \includegraphics[scale=1.05]{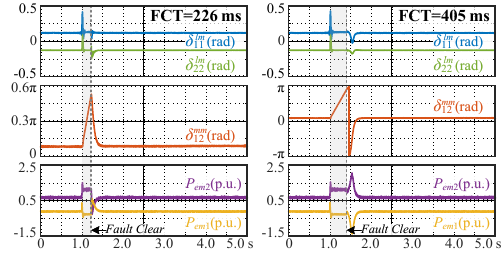}\\
    \vspace{-1pt}
    {\color{black}
    \hspace{0.8em}\ \footnotesize (a)\quad\quad\quad\quad\quad\quad\quad\quad\quad\quad\quad\quad\quad\quad\quad  \footnotesize (b)\\
    \vspace{1pt}}
    \includegraphics[scale=1.05]{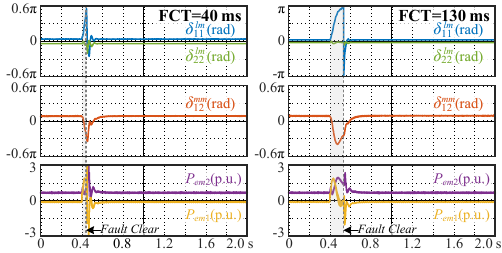}\\
    \vspace{-1pt}
    {\color{black}
    \hspace{0.8em}\ \footnotesize (c)\quad\quad\quad\quad\quad\quad\quad\quad\quad\quad\quad\quad\quad\quad\quad  \footnotesize (d)\\
    \vspace{-5pt}}
    \caption{Experimental results of the four-converter system. (a) Short-circuit fault at $\mathbf{Z}_g$ with FCT=226 ms. (b) Short-circuit fault at $\mathbf{Z}_g$ with FCT=405 ms. (c) Short-circuit fault at $\mathbf{Z}_1$ with FCT=40 ms. (d) Short-circuit fault at $\mathbf{Z}_1$ with FCT=130 ms. }
    \label{Fig10_Waveform_4M}
\end{figure}
\vspace{-2pt}
{\color{black}The four-converter system in Section~\ref{section_4b} is further tested under short-circuit faults occurring either at the midpoint of the impedance $\mathbf{Z}_g$ or at the midpoint of the impedance $\mathbf{Z}_1$. The corresponding estimated CCT results and their actual values for the two four-converter systems in \figref{Fig6_2A4M_Structure} are presented in TABLE~\ref{tabel_result_4M}, denoted as CCT$_A$ and CCT$_R$, respectively, where the CCT$_A$ results are obtained by identifying the intersection point between the fault-on trajectory and the ellipsoidal boundary shown in \figref{Fig7_2A4M_PhasePlane}. From TABLE~\ref{tabel_result_4M}, it can be clearly observed that the estimated CCT results are always smaller than the corresponding actual CCT results, thereby demonstrating the conservativeness and effectiveness of the proposed stability evaluation method.
}

When a short-circuit fault occurs at the midpoint of $\mathbf{Z}_g$, the power-sharing among the two converter groups is disrupted and the two groups diverge from the common speed. The corresponding time-domain waveforms of the four-converter system with optimized parameters under different FCTs of 226 ms and 405 ms are illustrated in \figref{Fig10_Waveform_4M}(a) and \figref{Fig10_Waveform_4M}(b). 
{\color{black} During short-circuit fault transients, the fault persists throughout the fault period marked by the gray area in \figref{Fig10_Waveform_4M}. The fault is then cleared by removing the short-circuit point at the fault-clearing moment marked by the gray dashed line, which coincides with the end of the fault duration.
}
It is evident that the system remains stable when the FCT is 226 ms as the post-fault angles converge to the SEP. As the FCT increases to 405 ms in \figref{Fig10_Waveform_4M}(b), the two GFM converters lose synchronization as $\delta_{12}^{mm}$ undergoes a one-cycle oscillation. Both the stability and the instability are correctly indicated by the proposed method, as the estimated value CCT$_A$ of 228 ms is larger than 226 ms while smaller than 405 ms.

When a short-circuit fault occurs at the midpoint of $\mathbf{Z}_1$, VSC1 becomes isolated from the rest of the system, resulting in deviations in the relative angle $\delta_{11}^{lm}$. Additionally, the power-sharing between the two areas is altered, and the relative angle between the two GFM converters also experiences a large-signal disturbance. The related transient waveforms of the four-converter system with optimized parameters are shown in \figref{Fig10_Waveform_4M}(c) and \figref{Fig10_Waveform_4M}(d), where two FCTs of 40 ms and 130 ms are examined. When the FCT equals 40 ms, the post-fault system remains stable as all relative angles return to the original SEP, and this outcome is correctly predicted by the estimated result, as CCT$_A$ equals 40 ms. For the case with the FCT of 130 ms, VSC1 and VSC2 lose synchronization as $\delta_{11}^{lm}$ undergoes a one-cycle oscillation, and the instability is also indicated by the proposed method, given that CCT$_A$ of 40 ms is shorter than 130 ms. These results further validate the effectiveness of the proposed method for stability evaluation of the four-converter systems. 

\vspace{-5pt}
\section{Conclusion}\label{section_6}
In this paper, a small-gain-like framework for large-signal stability evaluation of multi-converter systems is proposed. The large-signal angle dynamics is first modeled with a set of interconnected relative angle motions. Based on this, the small-gain-like property for multi-converter systems is established within certain relative angle limits, which allows the construction of a Lyapunov function that characterizes the dissipation capability of the whole system. Moreover, an ellipsoidal region within the angle-limit set is identified, and it serves as an effective estimate of the large-signal stability region of multi-converter systems, which offers strong potential for large-signal stability assessment and parameter design optimization of power systems. 
The proposed method is further applied to evaluate the large-signal stability of a paralleled system and a four-converter system, and their stability boundaries are directly solved and effectively validated through experiments.
{\color{black}
Future work will investigate possible refinements to the proposed framework by incorporating direction-dependent interaction effects and developing improved dissipation estimation strategies to further reduce conservativeness.
}

\ifCLASSOPTIONcaptionsoff
  \newpage
\fi
\vspace{-5pt}
\bibliographystyle{IEEEtran}
\bibliography{Reference}

\end{document}